%% file: main-popl.tex
\documentclass[acmsmall,screen]{acmart}

\usepackage{hhline}
\usepackage{wrapfig}
\usepackage{listings}
\usepackage{tikz,longtable,graphics,graphicx,color,caption,subcaption}
\usepackage[linesnumbered,algoruled,boxed,lined,noend]{algorithm2e}
\usepackage[strict]{changepage}
\usepackage{subfiles}
\usetikzlibrary{arrows,automata,positioning,calc}
\usepackage{adjustbox}
\tikzset{initial text={}}
\usetikzlibrary{decorations.markings}
\usepackage{multicol}
\usepackage{stmaryrd}
\usepackage{mathpartir}
\usepackage{amsthm,amsmath}
\usepackage{algorithmic}
\usepackage{url}
\usepackage{makecell}
\usepackage{enumitem}
\usepackage{apxproof}
\setitemize{label={--},itemsep=.5\parskip,partopsep=0pt,leftmargin=2em} %


\setcopyright{rightsretained}
\acmDOI{10.1145/3632864}
\acmYear{2024}
\copyrightyear{2024}
\acmSubmissionID{popl24main-p99-p}
\acmJournal{PACMPL}
\acmVolume{8}
\acmNumber{POPL}
\acmArticle{22}
\acmMonth{1}
\received{2023-07-11}
\received[accepted]{2023-11-07}
\begin{document}
\title{Regular Abstractions for Array Systems}

\author{Chih-Duo Hong}
\orcid{0000-0002-4064-8413}
\affiliation{%
  \institution{National Chengchi University}
  \city{Taipei}
  \country{Taiwan}
}
\email{chihduo@nccu.edu.tw}
\authornote{Corresponding author}

\author{Anthony W. Lin}
\orcid{0000-0003-4715-5096}
\affiliation{%
  \institution{University of Kaiserslautern-Landau}
  \city{Kaiserslautern}
  \country{Germany}
}
\affiliation{%
  \institution{MPI-SWS}
  \city{Kaiserslautern}
  \country{Germany}
}
\email{awlin@mpi-sws.org}

\include{macros}

\begin{abstract}
    \input{abstract}
\end{abstract}


\maketitle

\input{overview}

\input{illustrating-example}

\input{preliminaries}
\input{abstraction}
\input{algorithm}

\input{verification}

\input{case-studies}
\input{evaluation}

\input{related-work}

\input{conclusion}

\begin{acks}
We thank the anonymous reviewers for their insightful comments and corrections.
Chih-Duo Hong is partly supported by the National Science and Technology Council, Taiwan, under grant number 112-2222-E004-001-MY3. 
Anthony Lin is supported by \grantsponsor{501100000781}{European Research Council}{http://dx.doi.org/10.13039/501100000781} under European Union’s Horizon
2020 research and innovation programme (grant agreement no \href{https://doi.org/10.3030/101089343}{\grantnum{501100000781}{101089343}}).
\end{acks}

\bibliography{references}
\end{document}

%% file: macros.tex
\newlist{steps}{enumerate}{1}
\setlist[steps, 1]{leftmargin=1.2cm, label = Step \arabic*:}

\newcommand{\OMIT}[1]{}
\newcommand{\assign}[1]{\Sigma(#1)}
\newcommand{\ii}{\mathtt{i}}

\newcommand{\size}[1]{\ensuremath{\lvert #1 \rvert}}
\newcommand{\op}[1]{\operatorname{#1}}
\newcommand{\const}[1]{{\sf #1}}
\newcommand{\iset}[1]{\langle #1 \rangle}
\newcommand{\arr}[1]{\langle #1 \rangle}
\newcommand{\pred}[1]{\mathcal{#1}}

\newcommand{\NatInt}[2]{\ensuremath{\{#1,\ldots,#2\}}}
\renewcommand{\mod}[2]{\ensuremath{#1~{\rm mod}~#2}}
\newcommand{\floor}[1]{\lfloor #1 \rfloor}
\newcommand{\ceil}[1]{\lceil #1 \rceil}
\newcommand{\lowerbnd}[1]{\lfloor{#1}\rfloor}
\newcommand{\upperbnd}[1]{\lceil{#1}\rceil}
\newcommand{\norm}[1]{\lVert#1\rVert}
\newcommand{\nonneg}[1]{#1 \ge 1}

\newcommand{\voc}{\mathcal{V}}
\newcommand{\AP}{\mathsf{AP}}
\newcommand{\ACT}{\mathsf{ACT}}
\newcommand{\nondet}[2]{\mathsf{ndet}(#1,#2)}
\newcommand{\regularPTS}{\iset{S, P, r_+, \{r_a \}_{a\in\ACT}}}
\newcommand{\PTS}{\iset{S, P, +, \{\delta_a \}_{a\in\ACT}}}
\newcommand{\LTS}{\iset{S, \{\to_a \}_{a\in\ACT}}}
\newcommand{\TranSystemTriple}{\voc, \phi_I, \phi_T}
\newcommand{\TranSystem}{\ensuremath{(\TranSystemTriple)}}
\newcommand{\TranSystemSafe}{\ensuremath{(\TranSystemTriple, \phi_{bad})}}
\newcommand{\TranSystemFair}{\ensuremath{(\TranSystemTriple, \phi_{fair})}}

\newcommand{\UNIV}{\ensuremath{\textrm{FO}_\textrm{REG}}}
\newcommand{\args}[1]{\overline{#1}}
\newcommand{\ialphabet}{\ensuremath{\Sigma}}
\newcommand{\struct}{\mathfrak{S}}
\newcommand{\sem}[1]{\llbracket #1 \rrbracket}
\newcommand{\FO}[1]{\mathrm{FO}(#1)}
\newcommand{\MSO}[1]{\mathrm{MSO}(#1)}
\newcommand{\WMSO}[1]{\mathrm{WMSO}(#1)}
\newcommand{\zero}{\mathtt{0}}
\newcommand{\one}{\mathtt{1}}
\newcommand{\true}{\mathsf{true}}
\newcommand{\false}{\mathsf{false}}
\newcommand{\subst}[3]{#1\, [#2 / #3]}
\newcommand{\eqlen}{\ensuremath{eqL}}
\newcommand{\univ}{\,\mathfrak{U}\,}
\newcommand{\univDesc}{\langle \ialphabet^*\!,\: \preceq,\: {\eqlen,}\: \{\prec_a\}_{a \in \ialphabet}\rangle}
\newcommand{\lang}[1]{\mathcal{L}(#1)}
\newcommand{\rel}[1]{\sem{#1}}
\newcommand{\allows}{\,\rhd\,}
\newcommand{\bound}{n}
\newcommand{\theory}{\mathcal{T}}

\newcommand{\abs}{\kappa}

\makeatletter
\DeclareRobustCommand{\circle}{\mathord{\mathpalette\is@circle\relax}}
\newcommand\is@circle[2]{%
    \begingroup
    \sbox\z@{\raisebox{\depth}{$\m@th#1\bigcirc$}}%
    \sbox\tw@{$#1\square$}%
    \resizebox{!}{\ht\tw@}{\usebox{\z@}}%
    \endgroup
}
\makeatother
\newcommand{\Mydiamond}[1]{\ensuremath{\langle #1 \rangle}}
\def\G{\mathbf{G} \,}
\newcommand{\F}{\mathbf{F} \,}
\newcommand{\GF}{\mathbf{GF} \,}
\newcommand{\FG}{\mathbf{FG} \,}
\newcommand{\Until}{\,\mathbf{U}\,}
\newcommand{\Next}{\circle\,}

\newcommand{\Nat}{\mathbb{N}}
\newcommand{\Real}{\mathbb{R}}
\newcommand{\Rat}{\mathbb{Q}}
\newcommand{\Int}{{\mathbb{Z}}}

\newcommand{\T}{\mathsf{T}}
\newcommand{\N}{\mathsf{N}}
\newcommand{\W}{\mathsf{W}}
\newcommand{\R}{\mathsf{R}}
\renewcommand{\S}{\mathsf{S}}
\newcommand{\blank}{\mathtt{\#}}

\newcommand{\structA}{\mathfrak{A}}
\newcommand{\structB}{\mathfrak{B}}
\newcommand{\structW}{\mathfrak{W}}
\newcommand{\structT}{\mathfrak{T}}
\newcommand{\structM}{\mathfrak{M}}
\newcommand{\structR}{\mathfrak{R}}
\newcommand{\structH}{\mathfrak{H}}

\newcommand{\bbA}{\mathbb{A}}
\newcommand{\bbE}{\mathbb{E}}
\newcommand{\bbD}{\mathbb{D}}
\newcommand{\bbP}{\mathbb{P}}
\newcommand{\bbQ}{\mathbb{Q}}
\newcommand{\bbI}{\mathbb{I}}
\newcommand{\bbR}{\mathbb{R}}
\newcommand{\bbS}{\mathbb{S}}
\newcommand{\bbZ}{\mathbb{Z}}

\newcommand{\cA}{\mathcal{A}}
\newcommand{\cB}{\mathcal{B}}
\newcommand{\cC}{\mathcal{C}}
\newcommand{\cD}{\mathcal{D}}
\newcommand{\cI}{\mathcal{I}}
\newcommand{\cJ}{\mathcal{J}}
\newcommand{\cE}{\mathcal{E}}
\newcommand{\cF}{\mathcal{F}}
\newcommand{\cG}{\mathcal{G}}
\newcommand{\cV}{\mathcal{V}}
\newcommand{\cM}{\mathcal{M}}
\newcommand{\cQ}{\mathcal{Q}}
\newcommand{\cR}{\mathcal{R}}
\newcommand{\cT}{\mathcal{T}}
\newcommand{\cP}{\mathcal{P}}
\newcommand{\cL}{\mathcal{L}}
\newcommand{\cS}{\mathcal{S}}

\newcommand{\mT}{\mathfrak{T}}
\newcommand{\mS}{\mathfrak{S}}
\newcommand{\mU}{\mathfrak{U}}

\newcommand{\langA}{A}
\newcommand{\langB}{B}
\newcommand{\langT}{T}
\newcommand{\langI}{I}
\newcommand{\langR}{R}
\newcommand{\langE}{Z}

\newcommand{\concat}{\cdot}
\newcommand{\+}{\ensuremath{\!\cdot\!}}

\newcommand{\ModelRun}{\ensuremath{\pi}}
\newcommand{\empseq}{\ensuremath{\varepsilon}}
\newcommand{\tran}[1]{\ensuremath{\stackrel{#1}{\longrightarrow}}}

\newcommand{\MEM}{\mathit{Mem}}
\newcommand{\EQ}{\mathit{Equ}}

\newcommand{\Fp}{F_Y}
\newcommand{\myvec}[1]{
\begin{bmatrix}
	#1
\end{bmatrix}	
}

\newcommand{\toss}{{\sf toss}}
\newcommand{\tail}{{\sf tail}}
\newcommand{\head}{{\sf head}}
\newcommand{\reset}{{\sf \bot}}

\newcommand{\libalf}[0]{LibAlf}

\newcommand{\eval}[1]{\llbracket#1\rrbracket}

\newcommand{\constr}[1]{\mathsf{cstr}(#1)}

\def\qed {{%
        \parfillskip=0pt        %
        \widowpenalty=10000     %
        \displaywidowpenalty=10000  %
        \finalhyphendemerits=0  %
        \leavevmode             %
        \unskip                 %
        \nobreak                %
        \hfil                   %
        \penalty50              %
        \hskip.2em              %
        \null                   %
        \hfill                  %
        \qedsymbol\kern-.6pt%
        \par}}                  %
\def\qedsymbol{$\square$}

\newcommand{\note}[1]{{\color{red} #1}}
\newcommand{\sidechihduo}[1]{\todo[backgroundcolor=blue!20]{{\bf C} #1}}
\newcommand{\sideanthony}[1]{\todo[backgroundcolor=green!20]{{\bf A} #1}}

%% file: abstract.tex
Verifying safety and liveness over array systems is a highly challenging problem. Array systems naturally capture parameterized systems such as distributed protocols with an unbounded number of processes. Such distributed protocols often exploit process IDs during their computation, resulting in array systems whose element values range over an infinite domain. In this paper, we develop a novel framework for proving safety and liveness over array systems. The crux of the framework is to overapproximate an array system as a string rewriting system (i.e.~over a finite alphabet) by means of a new predicate abstraction that exploits the so-called indexed predicates. This allows us to tap into powerful verification methods for string rewriting systems that have been heavily developed in the last two decades or so (e.g.~regular model checking).
We demonstrate how our method yields simple, automatically verifiable proofs of safety and liveness properties for challenging examples, including Dijkstra's self-stabilizing protocol and the Chang-Roberts leader election protocol.

%% file: overview.tex
\section{Introduction}

Over the past few decades, extensive research efforts
(e.g.~\cite{ma2019i4,gurfinkel2016smt,mann2022counterexample,felli2021smt,cimatti2021universal})
have been devoted to the verification of \emph{array systems}.
Array systems are natural models of sequential programs
manipulating linear data structures such as arrays and lists.
In addition, they have also been used as convenient abstractions
of parameterized concurrent systems with local and shared variables
(cf.~\cite{alberti2017framework,sasha-book}).
Specifically, many distributed protocols require each process to maintain
a numerical local variable as its process identifier,
rendering these protocols suitable to be modeled as array systems.

Despite the amount of work on array systems verification, the problem
remains highly challenging. Difficulty arises from the following two aspects, among others:
(1) array elements range over an \emph{infinite domain} (e.g.~the set of integers), and
(2) many properties of interest require \emph{quantification} over the array elements.
For instance, the property ``array $a$ is sorted in ascending order''
is expressed by a universally quantified formula stating that
each element in the array $a$ is no smaller than its preceding elements.
For this reason, in order to verify the correctness of an array system,
it is often necessary to reason about quantified formulae,
which is generally undecidable over arrays (cf.~\cite{kroening-book,bradley-book}).

Owing to the undecidability of quantified array theories,
existing techniques for handling array systems employ a mixture of SMT, model checking, synthesis, and/or abstraction.
Most effort in developing SMT over arrays concentrates on providing more support of universal quantification \cite{mann2022counterexample,cimatti2021universal,ge2009complete,bradley2006,habermehl2008logic}.
To lift SMT over arrays to the verification of array systems, one may exploit model checking techniques like IC3/PDR (e.g.~\cite{ic3-array-bjorner,cimatti2016infinite,gurfinkel2018quantifiers}), backward reachability (e.g.~\cite{ranise2010backward,abdulla2009approximated}), synthesis (e.g.~SyGuS \cite{fedyukovich2019quantified}), eager abstraction (e.g.~\cite{mcmillan2018eager,mann2022counterexample}),
interpolation (e.g.~\cite{mcmillan2008quantified,hoenicke2018efficient,ghilardi2021interpolation}), or predicate abstraction
\cite{flanagan2002predicate,lahiri2004indexed,lahiri2007predicate}, among others. 
It should be remarked that, besides a handful of work (e.g.~\cite{padon2017reducing,padon2021temporal}), most existing work on array systems verification concerns \emph{safety verification}, with liveness still posing a significant challenge.

\paragraph{Contributions.}
This paper presents a new method for reasoning about array systems. We demonstrate its efficacy for verifying safety and liveness of challenging examples from array programs and distributed protocols, including Dijkstra's self-stabilizing protocol and the Chang-Roberts leader election protocol. The method has two main ingredients. Firstly, we provide a new \emph{predicate abstraction} (using \emph{indexed predicates}) that overapproximates an array system as a \emph{string rewriting system}.
String rewriting systems can be construed as array systems, whose array elements range over a \emph{finite domain} (a.k.a.~alphabet). The second ingredient answers why we use string rewriting systems as abstractions: this subclass of array systems has been more amenable to solutions than the general case of arrays. In fact,
several powerful verification methods over string rewriting systems (for \emph{both} safety and liveness, among others) have been developed in the last decade or so, which include methods in the framework \emph{regular model checking} \cite{RMC,bouajjani2004abstract,lin2022regular,Parosh12}.
This framework relies on the decidable first-order theory over words with prefix-of-relation, regular constraints, and the equal-length predicate (a.k.a.~the universal automatic structure \cite{blumensath2000automatic,benedikt2003definable}). The theory of universal automatic structure can be construed as a kind of array theory, whose elements range over a finite domain (e.g. bitvectors). In stark contrast to
decidable array theories like~\cite{habermehl2008logic,bradley2006}, however, the theory has no restrictions on the use of quantifiers, but restricts the manner in which the indices are related.
The set of solutions definable by such formulae are precisely those that are captured by synchronized automata running on tuples of words (a.k.a.~\emph{regular relations}).

\begin{table}
    \centering
    \caption{Comparison between our approach and the classical methods.
    The first column is the setting of the classical (indexed) predicate abstraction
    (e.g.~\cite{flanagan2002predicate,lahiri2007predicate,jhala2018predicate}).
    The second column provides our setting of indexed predicate abstraction (Section~\ref{sec:indexed-predicate-abstraction}).
    The last column is our approximation of indexed predicate abstraction using regular languages (Section~\ref{sec:regular-approximation}), which we refer to as
    \emph{regular abstraction}.
    }
    \scalebox{0.9}{
    \begin{tabular}{|l||c|c|c|}
         \hhline{-||---} 
         & \thead{Predicate\\Abstraction} & \thead{Indexed Predicate\\Abstraction} & \thead{Regular\\Abstraction}\\
         \hhline{=::===}
         {\small Abstract Set} & \thead{finite set of\\bitvectors} & \thead{infinite set of words\\over bitvectors} & \thead{finite-state\\automaton} \\
         \hhline{-||---}
         {\small Abstract Relation} & \thead{finite relation\\over bitvectors} & \thead{infinite relation\\over words} & \thead{finite-state\\transducer} \\
         \hhline{-||---}
         {\small Abstract System} & \thead{finite Boolean\\program} & \thead{infinite-state\\transition system} & \thead{regular\\transition system} \\
         \hhline{-||---}
    \end{tabular}}
    \label{fig:comparisons}
\end{table}

We now provide some details of our method. We use a variant of \emph{First-Order Linear Temporal Logic (FO-LTL)}~\cite{abadi1989power,hodkinson1999decidable} restricted to a quantified array theory,
in order to specify both an array system and properties to be verified.
This formalism may also be regarded as a variant of \emph{indexed LTL}
\cite{clarke1986reasoning,german1992reasoning}
with atomic propositions replaced by expressions from the array theory. Our predicate abstraction reduces this FO-LTL model checking problem over array systems to a verification problem over string rewriting systems that can be modeled in the framework of regular model checking. We discuss next our notion of predicate abstraction by means of indexed predicates. Indexed predicates have been used in previous work (e.g.~in \cite{flanagan2002predicate} for invariant generation and in \cite{lahiri2007predicate} for abstracting into a finite system), but never in the context of computing another infinite-state system that is more amenable to analysis. To this end, the right logic instantiating the indexed predicates has to be devised that suits our abstraction as string rewriting systems. In our case,
each \emph{indexed predicate} is an atomic formula (a.k.a.~atom) with a designated variable $i$, e.g., $a[i] \le a[n]$.
Given indexed predicates $\iset{P_1,\ldots,P_m}$, interpretations of an array formula $\phi$ are mapped into a word $w$ over the alphabet $\{0,1\}^m$ of $m$-bitvectors. This word $w$ summarizes the truth values of the predicates over all positions in the arrays interpreting $\phi$.

Although we have chosen to abstract away the actual array values by means of such indexed predicates, our abstraction maintains the concrete values of array sizes and index variables in the abstract domain. This strategy often allows us to compute a precise enough
abstraction based on a small set of predicates. Given an interpretation $\sigma$ of an array formula (as numbers and arrays), the set $\alpha(\sigma)$ of abstract values corresponding to $\sigma$ (as tuples of numbers and words over a finite alphabet) is computable.
We propose next to overapproximate sets and relations in the abstract domain with regular sets and relations.
We show that abstractions of quantified array formulae and array systems can be suitably approximated
using regular language as a symbolic representation.
This yields our notion of \emph{regular abstractions} for array formulae and systems.
Figure~\ref{fig:comparisons} compares this approach with the classical predicate abstraction.

We show that regular abstractions are sufficiently precise for verification purposes
when the quantified array formulae used in the FO-LTL specifications
belong to the so-called \emph{Singly-Indexed Array Logic (SIA)} \cite{habermehl2008logic}. To evaluate our approach,
we consider non-trivial safety and liveness properties
for several case studies,
including selection and merge sort algorithms,
Dijkstra's self-stabilizing algorithm~\cite{dijkstra1982self},
and the Chang-Roberts leader election algorithm~\cite{chang1979improved}.
We report promising experimental results for these case studies
using off-the-shelf regular model checkers as backend inference engines.

%% file: illustrating-example.tex
\section{An illustrative example}
\label{sec:illustrative-example}
Dijkstra's self-stabilizing algorithm~\cite{dijkstra1982self}
assumes a ring of $n \ge 2$ processes with identifiers $1,\ldots,n$
and $n$ local variables $x_1, \ldots, x_n$ over $\{0, \ldots,k-1\}$,
where $n$ and $k$ are finite but unbounded parameters satisfying $n \le k$. 
For $i \in \NatInt{1}{n}$,
Process $i$ is~\emph{privileged} if
(1) $i=1$ and $x_1 = x_n$, or
(2) $i>1$ and $x_i \neq x_{i-1}$.
At each time step, a privileged process, say Process $i$, is scheduled,
and the corresponding variable $x_i$ is updated by setting
(1) $x_i := x_i + 1~\mathrm{mod}~k$ when $i = 1$, or
(2) $x_i := x_{i-1}$ when $i>1$.
Process $i$ thus loses its privilege after the update.
Dijkstra's self-stabilizing algorithm can be construed as a kind of token-passing algorithm, whereby a process holds a token if and only if it is privileged.
The algorithm is self-stabilizing in the sense that eventually the system will contain a single token forever.

Let us identify each variable $x_i$ with an array\footnote{
In this paper, array indices always start from 1. We use $a[i]$ to denote the $i$-th element of array $a$,
and use $\size{a}$ to denote the size of $a$.
We often use $[a_1,\dots, a_n]$ to represent an unnamed array $a$ with $\size{a}=n$ and $a[i] = a_i$ for $i \in \NatInt{1}{n}$.
} value $a[i]$, and define a formula
\[
    \mathsf{Priv}(i) :=  (i = 1 \wedge a[i] = a[{n}]) \vee (i > 1 \wedge i \le n \wedge a[i] \neq a[i-1])
\]
indicating that Process $i$ is privileged.
One interesting property of Dijkstra's algorithm is that Process $1$ will be privileged infinitely often,
regardless of what the initial state is and how the privileged processes are scheduled.
To prove this,
we may use an index variable $pid$ to represent the scheduled process,
namely, $pid = 1$ means that Process $1$ is scheduled. %
Then the property ``Process 1 is privileged infinitely often''
can be expressed in LTL as $\GF \mathsf{Priv}(1)$.
We would like to prove that this property holds for Dijkstra's algorithm
under the scheduling assumption $\G \mathsf{Priv}(pid)$,
namely, under the assumption that the scheduler always selects a privileged process.

Inspired by the definition of $\mathsf{Priv}(i)$,
we may abstract the local state of each Process $i$ using predicates
$\iset{a[i] = a[n], a[i] = a[i-1]}$.
The induced abstract state space is then
$\Nat \times \Nat \times \Sigma^*$ with $\Sigma := \{\zero\zero,\one\zero,\zero\one,\one\one\}$.
Each abstract state $(pid,n,w)$ comprises the valuations of $pid$ and $n$,
and a word $w$ over the alphabet $\Sigma$.
The $i$-th letter of $w$ essentially corresponds to the \emph{predicate abstraction} of~$a$ at position~$i$.
For example, $w[2] = \zero\one$ means that $a$ satisfies $\neg (a[2] = a[n]) \wedge (a[2] = a[1])$.
The abstraction function $\alpha$ then maps a concrete state, say $s := (1,3,[1,2,1])$,
to a set of abstract states, say
$\alpha(s) = \{(1,3,[\one\zero,\zero\zero,\one\zero], (1,3,[\one\one,\zero\zero,\one\zero])\}$.
Observe that we may naively overapproximate the abstraction of (the concrete states satisfying) $\mathsf{Priv}(1)$
by $F := \Nat \times \Nat \times (\one\zero + \one\one)\!\cdot\!\Sigma^*$.
Here, $F$ is an overapproximation in the sense that 
$\alpha(s) \not\subseteq F$ implies $s \not\models \mathsf{Priv}(1)$,
whereas $\alpha(s) \subseteq F$ implies nothing.
(For example, given $s := (1,2,[1,2])$,
we can deduce $s \not\models \mathsf{Priv}(1)$
by observing that $\alpha(s) = \{ (1,2,[\zero\zero, \one\zero]), (1,2,[\zero\one, \one\zero]) \} \not\subseteq F$.)
Similarly, the abstraction of $\mathsf{Priv}(pid)$ may be overapproximated by
$S := \left(\{1\} \times \Nat \times (\one\zero + \one\one)\!\cdot\!\Sigma^*\right) \uplus \bigcup_{i \ge 2} \left(\{i\} \times \Nat \times \Sigma^{i-1}\!\cdot\!(\zero\zero + \one\zero)\!\cdot\!\Sigma^*\right)$. 
Thus, $\GF \mathsf{Priv}(1)$ holds for Dijkstra's algorithm under the scheduling assumption $\G \mathsf{Priv}(pid)$
if there does not exist a maximal abstract path~$\pi \in S^*$ satisfying $\FG (S \setminus F)$,
namely, $\pi$ consists of states in $S$ and eventually stays away from $F$ forever.

In this manner, we can reduce verification problems
of an array manipulating system to those of a string manipulating system,
for which effective techniques and highly optimized tools exist
(e.g.~\cite{abdulla2012regular,chen2017learning,lin2016liveness,klarlund2001mona,klarlund2002mona,fiedor2017lazy,schuppan2006liveness}).
Unfortunately, the approximations $S$ and $F$ given above are too rough to verify
the desired property, in the sense that they will produce spurious counterexample paths after the reduction. We shall later show that an appropriate choice of predicates and approximations will allow us to verify this property.

%% file: preliminaries.tex
\section{Preliminaries}
\label{sec:array-systems}

\noindent
\textbf{Notation.}
We adopt the standard notation of many-sorted logic
(see e.g.~\cite{van1994logic}), and
consistently use $\voc$ to denote a set of sorted first-order variables.
We use $\assign{\voc}$ to denote
the set of sort-consistent interpretations of $\voc$.
Given two interpretations
$\sigma_1 \in \assign{\voc_1}$ and
$\sigma_2 \in \assign{\voc_2}$
with $\voc_1 \cap \voc_2 =\emptyset$,
we use
$\sigma_1 \cdot \sigma_2$
to denote the interpretation
$\sigma \in \assign{\voc_1 \uplus \voc_2}$
such that $\sigma(v) = \sigma_1(v)$ for $v \in \voc_1$ and
$\sigma(u) = \sigma_2(u)$ for $u \in \voc_2$.
We write $v'$ for the primed version of variable $v$,
and define $\voc' := \{ v' : v \in \voc \}$.
Given $\sigma \in \assign{\voc}$, we use $\sigma'$ to denote
the interpretation $\{ v' \mapsto \sigma(v) : v \in \voc\}$ of $\voc'$.
Given a formula $\phi$, we use $\phi[t/x]$ to denote the formula
obtained by substituting $t$ for all free occurrences of $x$ in $\phi$.
Finally, for $\phi$ defined over variables $\voc$, we let
$\sem{\phi} := \{ \sigma \in \assign{\voc} : \sigma \models \phi \}$,
and write $\phi \equiv \psi$ when $\sem{\phi} = \sem{\psi}$.

\smallskip
\noindent
\textbf{Indexed array logic.}
An array theory typically combines two one-sorted theories:
an \emph{index theory}, which is used to express relations between indices,
and an \emph{element theory}, which is used to express properties of the array entries.
The \emph{indexed array logic} is a fragment of quantified array theory that
(i) uses Difference Arithmetic as the index theory, and
(ii) allows quantification only over index variables.
To make our expositions concrete, we shall focus on an instantiation
of indexed array logic that uses Difference Arithmetic as the element theory.
(However, we note that our abstraction techniques can be directly applied
to any indexed array logic with a decidable quantifier-free fragment,
which allows for a rich set of element theories such as Linear Integer Arithmetic (cf.~\cite{kroening-book}).)

Throughout this paper, we shall fix an indexed array logic $\theory$ with syntax given by
\begin{align*}
I & ::= I + n \mid i \mid n \mid \size{a} && \text{(index terms)}\\
E & ::= E + n \mid v \mid n \mid a[I] && \text{(data terms)}\\
F & ::= I \bowtie I \mid E \bowtie E \mid \neg F \mid F \vee F \mid F \wedge F \mid \exists i.\,F \mid \forall i.\,F && \text{(formulae)}
\end{align*}
where
$\bowtie\ \in \{=, \neq, \le, <, \ge, >\}$,
$n$ is a concrete integer,
$v$ is a data variable,
$i$ is an index variable,
$a$ is an array variable,
$\size{a}$ represents the array length of $a$,
and
$a[I]$ represents the array element of $a$ stored at position $I$.
A formula in form of $I \bowtie I$ or $E \bowtie E$ is called
an \emph{atomic formula}, or an \emph{atom} for short.
We extend the logic with logical connectives such as
$\Rightarrow$ (implies) and $\Leftrightarrow$ (if and only if)
in the usual way.
An interpretation $\sigma$ of a formula in the indexed array logic is intuitive: $\sigma$ maps each integer/index variable to an integer, and each array variable $a$ to a function $f_a: \mathbb{Z} \to \mathbb{Z}$. The notion of satisfaction of a formula with respect to this interpretation is now standard.
Although arrays are infinite, it is easy to confine an array $a$ to a finite array by focusing on the index range $\NatInt{1}{\size{a}}$. For example, a formula $\forall i.\,\phi(i,a)$ can simply be rewritten to $\forall i.\,1 \leq i \wedge i \le \size{a} \Rightarrow \phi(i,a)$.

\smallskip
\noindent
\textbf{First-order array system.}
A \emph{first-order array system}
$\structT := \TranSystem$ is a triple
where $\voc$ is a finite set of first-order variables,
and $\phi_I$ and $\phi_T$ are array formulae over
variables $\voc$ and $\voc \uplus \voc'$, respectively. 
The variables in~$\voc$ are grouped into \emph{state variables}
and \emph{system parameters}.
A state variable has either an index sort, an element sort, or an array sort.
A system parameter is a special index variable
whose value is determined in an initial state
and immutable during system execution.
For each array variable $a \in \voc$,
we identify a system parameter in $\voc$ with the array size $\size{a}$.
This effectively stipulates that the array size is not changeable after the array is allocated.

The semantics of an array system is determined by the array logic $\theory$
and interpretations for $\voc$. A \emph{state}
of an array system $\structT := \TranSystem$ is a sort-consistent interpretation
$s \in \assign{\voc}$ for the variables in $\voc$.
The formula $\phi_I$ specifies the set of initial states.
The formula $\phi_T$ specifies the system's evolution
by relating the current variables $\voc$ to their updated counterparts $\voc'$.
A (finite or infinite) sequence $\pi := s_0 \cdots s_n$ of states
with $n \in \Nat \cup \{\omega\}$
is a~\emph{path} if $s_0 \models \phi_I$
and $s_{i-1} \+ s_{i}' \models \phi_T$ for $i \in \NatInt{1}{n}$.
Since system parameters are immutable,
$\pi$ is a legitimate path only if
$s_{i-1}(v) = s_{i}(v)$ for all system parameters $v \in \voc$ and each $i \in \NatInt{1}{n}$.

\smallskip
\noindent
\textbf{A guarded modeling language.}
We describe a simple guarded command language
for specifying transitions of array systems.
While this language is less expressive than the full-fledged array logic,
it is more concise and understandable,
and is already expressive enough to capture all examples considered in this paper.
A guarded command is in form of \textit{guard $\allows$ update}.
A command is \emph{enabled} if its guard is satisfied.
At each time step, the system nondeterministically selects
an enabled command and updates the current state accordingly.
If no command is enabled, the system halts.
More precisely, a guarded command is in form of
$\phi \allows \beta_1,\dots, \beta_m$,
where $\phi$ is a (possibly quantified) array formula,
and each $\beta_i$ is an assignment in one of the two forms:
\begin{itemize}
    \item $x := I$ or $x := *$, where $x$ is an index variable,
    $I$ is an index term, and $*$ is a special symbol;
    \item $a[I] := E$ or $a[I] := *$, where $I$ is an index term, $E$ is a data term, and $*$ is a special symbol.
\end{itemize}
The assignments $\beta_1,\dots, \beta_m$ are executed in parallel.
The symbol ``$*$'' stands for a nondeterministic value,
which ranges over the index domain when the left-hand side is $x$,
and ranges over the data domain when the left-hand side is $a[I]$.
For simplicity, assignments to data variables are not directly supported by our modeling language.
This however is not an essential restriction, since
data variables can be formally translated to array elements stored at dedicated positions, thereby allowing assignments as array elements.

\begin{example}
\label{eg:dijkstra}
The transitions of Dijkstra's algorithm in the illustrative example may be given as
\begin{flalign*}
& pid = 1 \wedge a[pid] = a[{n}] \wedge a[pid] < {k}-1 \allows a[pid] := a[pid] + 1,\: pid := *\\
& pid = 1 \wedge a[pid] = a[{n}] \wedge a[pid] = {k}-1 \allows a[pid] := 0,\: pid := *\\
& pid > 1 \wedge pid \le n \wedge a[pid] \neq a[pid-1] \allows a[pid] := a[pid-1],\: pid := *
\end{flalign*}
The array system $\structT := \TranSystem$ of Dijkstra's algorithm is specified by
\begin{align}
\label{eqn:dijkstra-first-order-spec}
    \voc & := \{ pid, a, {n}, {k} \} \nonumber\\
    \phi_I & := 2 \le n \wedge {n} \le {k} \wedge (\forall i.\: 1\le i \wedge i \le n \Rightarrow 0 \le a[i] \wedge a[i] < {k}) \nonumber\\
    \phi_T & := \psi_1 \vee \psi_2 \vee \psi_3 \nonumber
\end{align}
where $pid$ is an index variable, $a$ is an array variable, ${n}, {k}$ are system parameters,
and $\size{a}$ is identified with $n$.
The transition formula $\phi_T$ is defined by three formulae $\psi_1$, $\psi_2$, and $\psi_3$, which are in turn derived from the three guarded commands listed above. For example, $\psi_1$ can be formally derived from the first command as $pid = 1 \wedge a[pid] = a[{n}] \wedge a[pid] < {k}-1 \wedge a'[pid] = a[pid]+1~\wedge$ $1 \le pid' \wedge pid' \le n \wedge (\forall i.\: 1\le i \wedge i \le n \wedge i \neq pid \Rightarrow a'[i] = a[i])$.
\qed
\end{example}

%% file: abstraction.tex
\section{Regular languages as abstractions}
\label{sec:regulation-abstraction}

In this section, we define a novel predicate abstraction of array formulae and array systems using indexed predicates. Since each array will be mapped to a word (over a finite alphabet) in the abstraction, we shall use regular languages to represent sets of words, and use regular relations to represent relations over words. We start by recalling a logical framework of regular relations.

\subsection{A First-Order Framework of Regular Relations}

Let $\Sigma_m := \{\zero,\one\}^m$ denote the set of bitvectors of length $m$.
For $m \ge 1$, we define a two-sorted structure\footnote{
{We note that this structure is equally expressive in first-order logic
to the so-called \emph{universal automatic structures}}~\cite{benedikt2003definable,blumensath2004finite,colcombet2007transforming}.
}
    $\struct_m := ( \Sigma_m^*,~\Nat,~succ,~pred,~\cdot[\cdot],~\size{\cdot},~\Delta_1, \ldots, \Delta_m )$,
where $\Sigma_m^* := \{ v_1\cdots v_n : n\ge 0$ and $v_i \in \Sigma_m$ for each $i \}$
are words over bitvectors; $succ, pred : \Nat \to \Nat$ are defined by
$succ(n) := n+1$ and $pred(n) := \max\{0, n - 1\}$, respectively;
$\size{\cdot} : \Sigma_m^* \to \Nat$ is word length
(i.e.~$\size{w}$ is the length of $w$);
$\cdot[\cdot] : \Sigma_m^* \times \Nat \to \Sigma_m$
offers letter-level access for words,
i.e., $w[i]$ is the $i$th letter of $w$ for $1 \le i \le \size{w}$;
each $\Delta_k \subseteq \Nat \times \Sigma_m^*$
offers bit-level access for words,
such that $(i, w) \in \Delta_k$ if and only if $1 \le i \le \size{w}$
and the $k$th bit of $w[i]$ is $\one$.
We shall use $n + k$ and $n - k$ as syntactic sugar for $succ^k(n)$ and $pred^k(n)$,
respectively.

Given two words $w_1 \in \Sigma_{m_1}^*$ and $w_2 \in \Sigma_{m_2}^*$, we define
the \emph{convolution} $w_1 \odot w_2$ of $w_1$ and $w_2$
as the word $w \in \Sigma_{m_1 + m_2}^*$
such that $\size{w} = \max\{\size{w_1},\size{w_2}\}$ and
\begin{equation}
    w[i] = \begin{cases}
        w_1[i] :: w_2[i], & 1 \le i \le \min\{\size{w_1},\size{w_2}\} ;\\
        \zero^{m_1} :: w_2[i], & \size{w_1} < i \le \size{w_2};\\
        w_1[i] :: \zero^{m_2}, & \size{w_2} < i \le \size{w_1},
    \end{cases}
\end{equation}
where $u :: v$ denotes the concatenation of the bitvectors $u$ and $v$.
Intuitively, $w_1 \odot w_2$ is obtained by juxtaposing $w_1$ and $w_2$ (in a left-aligned manner),
and padding the shorter word with enough $\zero$s.
For example, for two words $\zero\one, \zero\zero\zero\one$ over $\Sigma_{1}$,
$\zero\one \odot \zero\zero\zero\one$ is the word
$(\zero,\zero)(\one,\zero)(\zero,\zero)(\zero,\one)$ over $\Sigma_{2}$.
We lift $\odot$ to languages by defining
$L_1 \odot L_2 := \{ w_1 \odot w_2 : w_1 \in L_1,~w_2 \in L_2 \}$.

We say that a relation
is \emph{regular} if its language representation is regular under unary encoding.
Formally, suppose w.l.o.g.~that $R \subseteq \Nat^r \times (\Sigma_{m}^*)^l$.
We define the \emph{language representation} of $R$
as $\lang{R} := \{ w_1 \odot \cdots \odot w_{r+l} \in \Sigma_{r\,+\,l\cdot m}^* :
(w_1,\ldots,w_{r+l}) \in R \}$
by identifying each $n \in \Nat$ with $\one^n \in \Sigma_1^*$.
The following result states that a relation is regular if and only if it is first-order definable in $\struct_m$,
see e.g., \cite{benedikt2003definable,blumensath2004finite,colcombet2007transforming}.

\begin{proposition}
    \label{prop:logic-automata}
    Given a first-order formula $\phi$ defined over $\struct_m$,
    we can compute a finite automaton
    recognizing $\lang{\sem{\phi}}$.
    Conversely, given a finite automaton recognizing
    $\mathcal{L} \subseteq N_1 \odot \cdots \odot N_r \odot L_1 \odot \cdots \odot L_l$,
    where $N_i \subseteq \one^*$ and $L_i \subseteq \Sigma_m^*$ for each $i$,
    we can compute a first-order formula
    $\phi$ over $\struct_m$ such that $\mathcal{L} = \lang{\sem{\phi}}$.
\end{proposition}

\subsection{Abstraction of Array Formulae}
\label{sec:indexed-predicate-abstraction}
We use indexed predicates to express constraints on arrays.

\begin{definition}[Indexed predicate]
  An \emph{indexed predicate}
  is an atom (in the indexed array logic) containing a designated index variable $\ii$.
  A set of indexed predicates $\pred{P} := \iset{ P_1, \ldots, P_m }$
  is said to be defined over variables $\voc$ if
  each predicate $P_k \in \pred{P}$ is defined over variables $\voc \uplus \{\ii\}$.
\end{definition}

For convenience, we shall fix a set of variables $\voc$ and a set of indexed predicates
$\pred{P} := \iset{ P_1, \ldots, P_m }$ over variables $\voc$
throughout the rest of this section.
For each bitvector $v := b_1 \cdots b_m \in \Sigma_{m}$ comprising $m$ bits,
we define a quantifier-free array formula $\phi_v^{\pred{P}}$ over variables $\voc \uplus \{\ii\}$ as follows:
\begin{eqnarray}
\label{def:valuation-at-i}
    \phi_v^{\pred{P}} := \bigwedge\nolimits_{1 \le k \le m} \tilde{P_k},
    \ \mbox{where}\ \tilde{P_{k}} =\begin{cases}
    P_{k}, & b_k = \one ;\\
    \neg P_{k}, & b_k = \zero.
    \end{cases}
\end{eqnarray}
Formula $\phi_v^{\pred{P}}$ effectively regards each bit $b_k$ of the bitvector $v$
as the truth value of predicate $P_k$ at the parametric position $\ii$.
We furthermore use a word over bitvectors
to encode the truth values of the predicates ranging over an interval of positions:
for each word $w \in \Sigma_{m}^*$, we define an array formula
\begin{equation}
  \label{def:valuation}
  \Phi_{w}^{\pred{P}} := \bigwedge\nolimits_{1 \le i \le \size{w}} \phi_{w[i]}^{\pred{P}} [i/\ii].
\end{equation}
For example, if $\pred{P} = \iset{x[\ii] = x[n], x[\ii] \neq x[\ii-1]}$ and $w = [\one\zero,\zero\one]$, then
$\Phi_{w}^{\pred{P}}$ is $x[1]=x[n] \wedge \neg (x[1]\neq x[0]) \wedge \neg (x[2]=x[n]) \wedge x[2] \neq x[1]$.
By construction, for $\sigma \in \assign{\voc}$,
$\sigma \models \Phi_{w}^{\pred{P}}$ holds if and only if
for $i \in \NatInt{1}{\size{w}}$ and $k \in \NatInt{1}{m}$,
the truth value of $\sigma \models P_k[i/\ii]$ coincides with the $k$th bit
of the bitvector $w[i]$.

Fix an array system $\structT := \TranSystem$.
Let $\voc_{int}$ denote the set of parameters \emph{and} index variables in $\voc$,
and $\voc_{par}$ denote the set of parameters in $\voc$.
Let $r := \size{\voc_{int}}$ and $m := \size{\pred{P}}$.
We shall identify an interpretation $\sigma \in \assign{\voc_{int}}$
with an $r$-tuple $(\sigma(x_1),\ldots,\sigma(x_r)) \in \Nat^r$
where $x_i$ denotes the $i$th variable in $\voc_{int}$.
We define the \emph{concrete domain} of the system as $\assign{\voc}$,
and the \emph{abstract domain} as $\Nat^r \times \Sigma_{m}^*$.
Each abstract state
$(t, w) \in \Nat^r \times \Sigma_{m}^*$
comprises a tuple $t \in \Nat^r$ representing an interpretation of $\voc_{int}$,
and a word $w \in \Sigma_{m}^*$ encoding the valuations of $\pred{P}$
at positions $1,\dots,\size{w}$.

Given a set $S_C \subseteq \assign{\voc}$ of concrete states, we define
\begin{align*}
    \alpha^{\pred{P}}(S_C) & :=
        \{ (s_{int}, w) \in \Nat^r \times \Sigma_{m}^* :
        \mbox{there exists $s \in S_C$ s.t.}~\size{w}=\max_{x \in \voc_{par}} s(x)\ \mbox{and}\ s \models \Phi_{w}^{\pred{P}}\,\}.
\end{align*}
Here, $s_{int}$ denotes the interpretation obtained by restricting $s$ to $\voc_{int}$.
Given a set $S_A \subseteq \Nat^r \times \Sigma_{m}^*$ of abstract states, we define
\begin{align*}
    \gamma^{\pred{P}}(S_A) & := \{ s \in \assign{\voc} : 
     \mbox{for all $w \in \Sigma_{m}^*$ s.t.}~\size{w}=\max_{x \in \voc_{par}} s(x),\ s \models \Phi_{w}^{\pred{P}}\ \mbox{implies}\ (s_{int}, w) \in S_A \}.
\end{align*}
{Notice that $\alpha^{\pred{P}}(S_C)$ and $\gamma^{\pred{P}}(S_A)$ are computable
when $S_C$ and $S_A$ are finite.}
To simplify the notation,
we may omit the superscript ${\pred{P}}$ from $\alpha^{\pred{P}}$ and $\gamma^{\pred{P}}$ when it is clear from the context.
Also, for a single state $s$,
we shall write $\alpha(s)$ and $\gamma(s)$ instead of $\alpha(\{s\})$
and $\gamma(\{s\})$ when there is no danger of ambiguity.
Finally, for a formula $\phi$,
we shall write $\alpha(\phi)$ and $\gamma(\phi)$ instead of $\alpha(\sem{\phi})$ and $\gamma(\sem{\phi})$.

\begin{example}
Recall the illustrative example in Section~\ref{sec:illustrative-example}. With indexed predicates
$\pred{P} := \iset{a[\ii] = a[n], a[\ii] = a[\ii-1]}$, we can define $\alpha$ and $\gamma$ as
\begin{align*}
    \alpha(S_C) & := \{ (pid,n,w) \in \Nat^2 \times \Sigma_2^* : \text{there exists}~(pid,n,a) \in S_C~\text{s.t.}~\size{w}=n~\text{and}~(pid,n,a) \models \Phi_w^{\pred{P}} \};\\
    \gamma(S_A) & := \{ (pid,n,a) \in \Nat^2 \times \Int^\omega : \text{for all}~w \in \Sigma_2^n,~(pid,n,a) \models \Phi_w^{\pred{P}}~\text{implies}~(pid,n,w)\in S_A \}.
\end{align*}
For instance, given a concrete state $s := (1,3,[1,2,1])$, we can compute
$\alpha(s) = \{(1,3,[\one\zero,\zero\zero,\one\zero]),$ $(1,3,[\one\one,\zero\zero,\one\zero])\}$
and $\gamma(\alpha(s)) = \{ (1,3,a) : a \in \Int^\omega,~a[1] = a[3] \neq a[2] \}$. %
\qed
\end{example}

\begin{lemma}
\label{lemma:dist-law}
For a countable collection $\{S_i : i \in I\}$ of concrete state sets,
it holds that $\alpha(\bigcup_{i \in I} S_i) = \bigcup_{i \in I} \alpha(S_i)$
and $\alpha(\bigcap_{i \in I} S_i) \subseteq \bigcap_{i \in I} \alpha(S_i)$.
\end{lemma}
\begin{proof}
Note that $\alpha$ distributes over union by definition, i.e., $\alpha(S) = \bigcup_{s\in S} \alpha(s)$.
Also, note that $s \in \bigcap_{i \in I} S_i$ implies $\alpha(s) \subseteq \bigcap_{i \in I}\alpha(S_i)$. Thus,
$\alpha(\bigcap_{i \in I} S_i) = \bigcup \{ \alpha(s) : s \in \bigcap_{i\in I} S_i \} \subseteq \bigcap_{i \in I} \alpha(S_i)$.
\end{proof}

\begin{lemma}
\label{lemma:galois}
For any abstract state set $S_A$, it holds that $\gamma(S_A) = \{ s : \alpha(s) \subseteq S_A \}$.
\end{lemma}

\begin{proof}
   It suffices to fix a concrete state $s$ and
   show that $s \in \gamma(S_A) \iff \alpha(s) \subseteq S_A$.
   For the ``$\Rightarrow$'' direction, assume that $s \in \gamma(S_A)$, and
   let $(s_{int},w)$ be an abstract state in $\alpha(s)$. %
   Then we have $s \models \Phi_{w}^{\pred{P}}$~by the definition of $\alpha$,
   which implies that $(s_{int}, w) \in S_A$ by the definition of $\gamma$.
   Since $(s_{int}, w)$ is arbitrary, %
   we conclude that $\alpha(s) \subseteq S_A$.
   For the ``$\Leftarrow$'' direction, assume that $\alpha(s) \subseteq S_A$.
   Let $W_s := \{ w \in \Sigma_{m}^* : \size{w}=\max_{x \in \voc_{par}} s(x) \}$.
   By the definition of $\alpha$, $S_A$ contains
   all abstract states $(s_{int}, w)$ such that $w \in W_s$
   and $s \models \Phi_{w}^{\pred{P}}$.
   However, this implies that for any $w \in W_s$, if
   $s \models \Phi_{w}^{\pred{P}}$ holds,
   then $(s_{int}, w) \in S_A$ also holds.
   It follows that $s \in \gamma(S_A)$ by the definition of $\gamma$.
\end{proof}

The following result %
is a consequence of Lemma~\ref{lemma:dist-law}, Lemma~\ref{lemma:galois}, and Proposition 7 of \cite{cousot1977abstract}.
We provide a proof here for the sake of completeness.

\begin{proposition}
    \label{prop:galois}
    For any concrete state set $S_C$ and abstract state set $S_A$, it holds that
    $$\alpha(S_C) \subseteq S_A \iff S_C \subseteq \gamma(S_A).$$
\end{proposition}
\begin{proof}
    For the ``$\Rightarrow$'' direction, assume that $\alpha(S_C) \subseteq S_A$.
    Since $\alpha$ is monotonic,
    i.e., $E\subseteq F$ implies $\alpha(E) \subseteq \alpha(F)$,
    we have $\alpha(s) \subseteq S_A$ for every $s \in S_C$.
    By Lemma~\ref{lemma:galois}, we have $\gamma(S_A) = \{s : \alpha(s) \subseteq S_A\}$.
    It then follows that $S_C \subseteq \gamma(S_A)$.
    For the ``$\Leftarrow$'' direction, assume that $S_C \subseteq \gamma(S_A)$.
    Again by Lemma~\ref{lemma:galois}, we have $\gamma(S_A) = \{s : \alpha(s) \subseteq S_A\}$.
    The assumption $S_C \subseteq \gamma(S_A)$ thus implies that
    $\alpha(s) \subseteq S_A$ for every $s \in S_C$.
    This in turn implies that $\alpha(S_C) \subseteq S_A$
    by Lemma~\ref{lemma:dist-law}.
\end{proof}

\subsection{Abstract Safety and Liveness Analysis}
\label{sec:abstract-analysis}
For an array system $\structT := \TranSystem$,
applying predicate abstraction for safety and liveness verification
involves performing reachability analysis for the abstract system of $\structT$.
In each step of the analysis, we concretize the current
set of reachable abstract states via the concretization operation $\gamma$,
apply the concrete next-step function
\begin{equation}
    N_C(S_C) := \{ t \in \assign{\voc} : \mbox{there exists $s \in S_C$ such that $s\+t' \models \phi_T$} \}
\end{equation}
in the concrete state space,
and map the result back to the abstract state space via the abstraction operation $\alpha$.
We can view this process as performing
state exploration in an abstract transition system
with a set $\alpha({\phi_I})$ of initial states
and an abstract next-step function
\begin{equation}
\label{def:abstract-next-step-function}
    N_A(S_A) := \alpha(N_C(\gamma(S_A))).
\end{equation}
Proposition~\ref{prop:galois}, together with the fact that $\alpha$, $\gamma$, and $N_C$ are monotonic,
establishes the soundness of abstract reachability analysis using $N_A$ (cf.~\cite{cousot1977abstract}).
For example,
let $\Pi(I, T)$ denote the set of paths starting from states in $I$ through transition relation $T$, and let $st(\pi)$ denote the set of states in path $\pi$.
Given a set $Init$ of initial states and a set $Bad$ of bad states, one can show (e.g.~by induction on the length of a counterexample path) that the safety property
\begin{equation}
    \label{formula:safety}
    \forall \pi \in \Pi(Init, N_C).~st(\pi) \cap Bad = \emptyset
\end{equation}
holds in the concrete domain if
$\forall \pi \in \Pi(\alpha(Init), N_A).~st(\pi) \cap \alpha(Bad) = \emptyset$
holds in the abstract domain.
Similarly, given a set $Fin$ of final states, the liveness property
\begin{equation}
\label{formula:liveness}
    \forall \pi \in \Pi(Init, N_C).~ st(\pi)\cap Fin \neq \emptyset
\end{equation}
holds in the concrete domain if
$\forall \pi \in \Pi(\alpha(Init),N_A).~st(\pi) \cap \alpha(Fin^\mathtt{c})^\mathtt{c} \neq \emptyset$
holds in the abstract domain,%
where $S^\mathtt{c}$ denotes the complement of the concrete (resp.~abstract) state set $S$
with respect to the concrete (resp.~abstract) state space.
These facts allow us to reduce the verification of concrete safety/liveness properties to that of their counterparts in the abstract domain.

\subsection{Approximation with Regular Languages}
\label{sec:regular-approximation}
Since indexed array logic is undecidable,
the abstraction $\alpha({\phi})$ of an indexed array formula $\phi$ is generally not computable.
Therefore, we propose to overapproximate these formulae in the abstract domain using regular languages.
As before, 
fix a set $\voc$ of variables and
a set $\pred{P}$ of indexed predicates over $\voc$,
and let $r := \size{\voc_{int}}$ and $m := \size{\pred{P}}$.
Then $\voc$ and $\pred{P}$ induce
an abstract domain $\Nat^r \times \Sigma_{m}^*$.
Recall that we regard a relation $R$ as regular when $R$'s language representation $\lang{R}$ is regular.

\begin{definition} [Regular abstraction for state formulae]
    Let $\phi$ be an indexed array formula over variables $\voc$.
    Then a tuple $({R}, \pred{P})$ is a \emph{regular abstraction} of $\phi$ if
    ${R} \subseteq \Nat^r \times \Sigma_{m}^*$ is a regular relation such that
    $\alpha^{\pred{P}}({\phi}) \subseteq {R}$.
\end{definition}

For $\pred{P} := \iset{ P_1, \ldots, P_m }$
over $\voc$, we use $\pred{P} * \pred{P}' :
= \iset{ P_1, \ldots, P_m, P_{m+1}, \dots, P_{2m} }$
to denote the set of predicates over $\voc \uplus \voc'$
such that for $k \in \NatInt{1}{m}$,
$P_{m+k}$ is obtained from $P_k$
by replacing $v$ with $v'$ for each variable $v \in \voc$.
Recall that $s\+ t$ denotes the union of interpretations $s$ and $t$ over disjoint variables,
and that $u \odot v$ denotes the convolution of words $u$ and $v$.

\begin{definition}[Regular abstraction for transition formulae]
    Let $\phi$ be an indexed array formula over variables $\voc \uplus \voc'$.
    Then a tuple $({R}, \pred{P})$ is a \emph{regular abstraction} of $\phi$ if
    ${R} \subseteq (\Nat^r \times \Sigma_{m}^*) \times (\Nat^r \times \Sigma_{m}^*)$
    is a regular relation such that $\alpha^{\pred{P} * \pred{P}'}({\phi}) \subseteq \tau({R})$,
    where $\tau$ %
    is an isomorphic mapping defined by
    $\tau((s, u), (t, v)) = (s \+ t, u \odot v)$.
\end{definition}

A regular abstraction $({R},\pred{P})$ of a transition formula
induces a next-step function $N_R$, defined by $N_R(S_A) := \{ t :$ there exists $s \in S_A$ such that~$(s,t) \in R\}$, over the abstract state sets.
It can be shown that $N_R$ provides an \emph{abstract interpretation}~\cite{cousot1977abstract} of the concrete array system.

\begin{theorem}
    Let $\structT := \TranSystem$ be an array system and $\pred{P}$ be a set of indexed predicates over $\voc$.
    If $({R}, \pred{P})$ is a regular abstraction of $\phi_T$,
    then $N_R$ provides an abstract interpretation of $\:\structT$.
\end{theorem}
\begin{proof}
    To show that $N_R$ is an abstract interpretation, we need to check that
    (1) $N_R$ is monotonic;
    (2) $N_R$ is null-preserving, namely, $N_R(\emptyset) = \emptyset$;
    (3) $N_R$ simulates $N_C$ w.r.t.~a simulation relation defined by $\alpha$, namely, $\alpha(N_C(S_C)) \subseteq N_R(\alpha(S_C))$.
    The first two conditions directly follow from the definition of regular abstraction,
    so it suffices to check the last condition.
   Suppose first that $S_C = \emptyset$.
   Since $N_C(\emptyset) = \alpha(\emptyset) = N_R(\emptyset) = \emptyset$, the condition holds trivially.
   Now suppose that $S_C \neq \emptyset$ and consider an arbitrary state $s \in S_C$.
   Observe that
   \begin{align*}
       \alpha(N_C(s)) & = \alpha(\{ t \in \assign{\voc} : s\+t' \models \phi_T \}) \\
            & = \{ (t_{int},v) : s\+t' \models \phi_T,\ \size{v}=\max\nolimits_{x \in \voc_{par}} t(x),\ t\models \Phi_v^\pred{P}\, \}\\
            & \subseteq {N_R}(\{ (s_{int}, u) : \size{u}=\max\nolimits_{x \in \voc_{par}} s(x),\ s\models \Phi_u^\pred{P}\, \})\\
            & = {N_R}(\alpha(s)),
   \end{align*}
   where the inclusion in the third line follows from the definition of regular abstraction
   and the stipulation that parameters are immutable.
   Note that $N_C$, $N_R$, and $\alpha$ are all distributive over union. %
   Since $s$ is arbitrary in $S_C$, we have 
   $\alpha(N_C(S_C)) = \bigcup_{s \in S_C} \alpha(N_C(s)) \subseteq \bigcup_{s \in S_C} {N_R}(\alpha(s)) = {N_R}(\alpha(S_C))$.
\end{proof}

Finally, a regular abstraction can be computed by composing smaller regular abstractions.
This fact is an immediate consequence of Lemma~\ref{lemma:dist-law}.

\begin{proposition}
    \label{prop:composition}
    Let $\phi_1,\dots,\phi_n$ be indexed array formulae, and
    $({R}_i,\pred{P})$ be a regular abstraction of $\phi_i$ for $i \in \NatInt{1}{n}$.
    Then $(\bigcup_{i=1}^n R_i, \pred{P})$ is a regular abstraction of $\bigvee_{i=1}^n \phi_i$,
    and $(\bigcap_{i=1}^n R_i, \pred{P})$ is a regular abstraction of $\bigwedge_{i=1}^n \phi_i$.
\end{proposition}

%% file: algorithm.tex
\section{Computation of regular abstractions}
\label{sec:computation-of-regular-abstractions}

We have proposed to overapproximate indexed predicate abstractions
with regular abstractions. However, it remains unclear how to find
such approximations that are sufficiently precise for a verification task.
In this section, we address this issue by providing an automated procedure to
compute regular abstractions for a fragment of
array logic called \emph{singly indexed array formulae}.
In a nutshell, give such a formula $\phi := \delta x_1 \cdots \delta x_n.\,\varphi$,
we compute an overapproximation of $\phi$ by replacing the matrix $\varphi$ of $\phi$
with a quantifier-free formula $\Tilde{\varphi}$
such that $\varphi \Rightarrow \Tilde{\varphi}$ is valid.
We shall choose $\Tilde{\varphi}$ in such a way
that the resulting formula $\delta x_1 \cdots \delta x_n.\,\Tilde{\varphi}$,
called an~\emph{abstraction formula},
can be faithfully encoded as a first-order formula over an automatic structure.
This encoding is equally expressive to regular languages,
and can be conveniently manipulated using tools such as Mona~\cite{klarlund2001mona,klarlund2002mona}
and Gaston~\cite{fiedor2017lazy} for abstract analysis.

\subsection{Singly Indexed Array Formulae}
An atomic indexed array formula is a~\emph{data expression}
if it contains an array variable; otherwise it is an~\emph{index expression}.
A~\emph{singly indexed array (SIA)} formula
is an indexed array formula in which every data expression
contains \emph{at most} one quantified index variable.
The satisfiability problem of SIA formulae
is already undecidable when the element theory is instantiated
to Difference Arithmetic~\cite[Lemma~5]{habermehl2008logic}.
Note that an SIA formula can be syntactically rewritten to a
logically equivalent SIA formula where every data expression has
\emph{exactly} one quantified index variable ---
hence the name ``singly indexed''. For example,
a data expression $a[1] \neq b[{n}]$ in an SIA formula can be
replaced with $\exists i.\, i = 1 \wedge a[i] \neq b[n]$
to obtain an equivalent SIA formula.
We shall assume an SIA formula to have this form
when we are computing a regular abstraction for the formula.

\subsection{A Constraint-Based Abstraction Procedure}
\label{sec:abstraction-procedure}

We need a couple more definitions before we introduce our abstraction method.
As before, fix a set $\voc$ of variables and %
a set $\pred{P} := \iset{ P_1, \ldots, P_m }$ of indexed predicates over $\voc$. %
A formula over $\voc$ is \emph{expressible in $\pred{P}$}
if each atom of the formula is either an index expression
or an expression of form $P[t/\ii]$ or $\neg P[t/\ii]$,
where $t$ is an index term and $P$ is a predicate in $\pred{P}$.
For a formula $\phi$ expressible in $\pred{P}$,
we shall use $\abs^\pred{P}(\phi)$ to denote the formula obtained by
replacing each atom $P_k[t/\ii]$ (resp.~$\neg P_k[t/\ii]$) in $\phi$
with $\psi \Rightarrow \Delta_k(t,w)$ (resp.~$\psi \Rightarrow \neg\Delta_k(t,w)$),
where $\psi$ encodes the boundary check for $P_k[t/\ii]$.
For example, suppose that $\phi = \forall i.\:a[i] \le b[p]$
and that $P_1 = a[\ii] > b[p]$ is a predicate in $\pred{P}$. Then
$
    \abs^\pred{P}(\phi) = \forall i.\:(1 \le i \wedge i \le \size{a} \wedge 1 \le p \wedge p \le \size{b} \Rightarrow \neg \Delta_1(i, w)).
$
Note that $\abs^\pred{P}(\phi)$ is defined over structure $\struct_m$ and variables $\voc_{int}\uplus\{w\}$.
Specifically, the free variables of $\abs^\pred{P}(\phi)$ comprise
the index variables and parameters of $\phi$, as well as a fresh word variable $w$ over domain $\Sigma_m^*$.
We may simply write $\abs^\pred{P}(\phi)$ as $\abs(\phi)$ when $\pred{P}$ is clear from the context.

Now we are ready to describe our abstraction method.
Consider a singly indexed formula $\phi$ over $\voc$.
Suppose w.l.o.g.~that the matrix of $\phi$ is given in disjunctive normal form
$\bigvee_j (g_j \wedge h_j)$,
where each $g_j$ is a conjunction of index expressions,
and each $h_j$ is a conjunction of data expressions.
Our procedure computes a regular abstraction for $\phi$ in two steps:
\begin{steps}
    \item Replace each $h_j$ in $\phi$ with a formula $\constr{h_j}$,
    where $\constr{\psi}$ denotes the conjunction of disjunctive clauses defined by
    \begin{align}
    \label{def:abstraction-constraint}
        \constr{\psi} := \bigwedge \left\{C \subseteq A(\psi,\pred{P}) : \,\psi \Rightarrow C~\,\text{is valid in}~\theory\,\right\},
    \end{align}
    and
    $A(\psi,\pred{P}) := \{ P_k[t/\ii], \neg P_k[t/\ii] : P_k \in \pred{P}$, and $t$ is an index term of $\psi$ or $\pred{P}\}$.
    Intuitively, $\constr{\psi}$ attempts to express in $\pred{P}$ the necessary conditions for
    a concrete state to satisfy $\psi$. Note that $A(\psi,\pred{P})$ is finite,
    and hence $\constr{\psi}$ is computable by our assumption on $\theory$.
    \item Let $\phi^*$ denote the formula obtained by
    replacing each $h_j$ with $\constr{h_j}$.
    This formula $\phi^*$, which we shall refer to as the \emph{abstraction formula} of $\phi$,
    is expressible in $\pred{P}$.
    We then define $\phi^\pred{P} = \phi^\pred{P}(\sigma,w) := \abs^\pred{P}(\phi^*)$,
    and use $(\sem{\phi^\pred{P}}, \pred{P})$ as a regular abstraction for $\phi$.
\end{steps}

Note that when computing $\constr{\psi}$, we can exclude a clause
$C$ if $\abs^\pred{P}(C)$ is valid in $\struct_m$
(e.g.~when $\{f,\neg f\}\subseteq C$ for some atom $f$),
since such a clause imposes no restrictions on the abstract states.
Furthermore, it suffices to consider the \emph{minimal} clauses $C$
satisfying $\psi \Rightarrow C$.
Thus, we can reduce the computation of $\constr{\psi}$ to enumerating the
\emph{minimal unsatisfiable cores (MUCs)} of $\psi \wedge \bigwedge A(\psi,\pred{P})$,
as negating such a core essentially leads to a minimal feasible clause of $\constr{\psi}$.
There are rich tool supports for MUC enumeration in the literature,
see~e.g.~\cite{liffiton2016fast,bendik2020must,bendik2018recursive}.
In practice, it is often possible to compute $\constr{\psi}$ incrementally
and obtain a precise enough abstraction without including all MUCs
in the abstraction formula.
We shall discuss these optimizations in Section~\ref{sec:evaluation}.

\begin{example}
    \label{ex:state-formula-abstraction-2}
    As an illustration, let us use our abstraction method to prove that 
    there always exists at least one privileged process in Dijkstra's self-stabilizing algorithm (cf.~Example~\ref{eg:dijkstra}). Consider the singly indexed array formula
    \[
        \phi := \forall i.\: (i=1 \wedge i < n \wedge a[i] \neq a[n]) \vee (i>1 \wedge i\le n \wedge a[i] = a[i-1]) \vee (i>n \wedge n \ge 2),
    \]
    which expresses that no process is privileged in a system containing $n \ge 2$ processes.
    We shall compute a regular abstraction of $\phi$ using
    indexed predicates $\pred{P} := \iset{a[\ii] = a[n],\, a[\ii] = a[\ii-1]}$, which comprises the atomic formulae of $\phi$.
    The first step computes the abstraction formula $\phi^*$ as
    \[
        \forall i.\: (i=1 \wedge i < n \wedge \constr{a[i] \neq a[n]}) \vee (i>1 \wedge i\le n \wedge \constr{a[i] = a[i-1]}) \vee (i>n \wedge n \ge 2).
    \]
    The nontrivial minimal feasible clauses for $\constr{a[i] \neq a[n]}$
    are $\{a[n]=a[n]\}$, $\{a[i] \neq a[n]\}$, and $\{a[i]\neq a[i-1], a[i-1]\neq a[n]\}$.
    (We call a clause $C$ trivial if $\abs(C)$ is valid in $\struct_m$.)
    Similarly, the nontrivial minimal feasible clauses for $\constr{a[i] = a[i-1]}$
    are $\{a[n] = a[n]\}$, $\{a[i] = a[i - 1]\}$, $\{a[i] = a[n], a[i - 1] \neq a[n]\}$, and $\{a[i] \neq a[n], a[i - 1] = a[n]\}$.
    From these clauses, we then compute a regular abstraction of $\phi$ as
    $\phi^\pred{P} = \phi^\pred{P}(n,w) := \abs(\phi^*) =
    \forall i.\: (i=1 \wedge i < n \wedge \abs(\constr{a[i] \neq a[n]})) \vee (i>1 \wedge i\le n \wedge \abs(\constr{a[i] = a[i-1]})) \vee (i>n \wedge n \ge 2).
    $
    For example,
    \begin{align*}
        \abs(\constr{a[i] \neq a[n]}) & =
        \abs(a[n]=a[n] \wedge a[i] \neq a[n] \wedge (a[i]\neq a[i-1] \vee a[i-1]\neq a[n])) \\
        & = (1 \le n \wedge n \le n \Rightarrow \Delta_1(n,w)) \wedge
             (1 \le i \wedge i \le n \Rightarrow \neg\Delta_1(i, w))~\wedge \\
        & \ \quad ((1 \le i \wedge i \le n \Rightarrow \neg\Delta_2(i, w)) \vee (1 \le i-1 \wedge i-1 \le n \Rightarrow \neg\Delta_1(i-1, w))).
    \end{align*}
    We can effectively check that $\phi^\pred{P}$ is unsatisfiable in ${\struct_2}$ by Proposition~\ref{prop:logic-automata}.
    Since regular abstractions are overapproximations,
    this implies that $\phi$ is unsatisfiable in $\theory$.
    In this way, with the provided indexed predicates,
    our tool can prove that $\phi$ is unsat in a second,
    while solvers like Z3 and cvc5 fail to handle the formula (i.e.~both of them output ``unknown'' for checking satisfiability of $\phi$.)
    \qed
\end{example}

\paragraph{Transition formulae}

Given a singly indexed transition formula $\phi$ over variables $\voc \uplus \voc'$,
we can compute a regular abstraction for $\phi$ by first extending the predicates
$\pred{P}$ to $\pred{P} \uplus \pred{P}'$,
where $\pred{P}'$ is obtained from $\pred{P}$ by replacing each $x\in \voc$ with $x'$.
We then compute an abstract transition formula
$\phi^\pred{P}((\sigma,w),(\sigma',w'))$ analogously to how we compute an abstract state formula $\phi^{\pred{P}}(\sigma,w)$.
However, if we only consider transition formulae induced by our guarded command language,
we can compute an abstract transition formula directly from a command as follows.
Define $\mathsf{copy\_except}(x) := \bigwedge\nolimits_{z \in \voc\setminus\{x\}} (z' = z)$.
Consider w.l.o.g.~a guarded command $\psi \allows \beta$, where $\beta$
is in form of either $x := e$ or $a[t] := e$.
For $x := e$, the corresponding transition formula is
\begin{equation}
    \label{eq:trans1}
    \psi \wedge \mathsf{copy\_except}(x) \wedge x' = e.
\end{equation}
For $a[t] := e$, the corresponding formula is
\begin{equation}
    \label{eq:trans2}
    \psi \wedge \mathsf{copy\_except}(a) \wedge a' = a\{t \leftarrow e\},
\end{equation}
where $a\{t \leftarrow e\}$ is a notation
from the theory of extensional arrays~\cite{mccarthy1993towards},
denoting the array obtained by assigning $e$ to $a[t]$.
(We may set $e$ to a fresh variable when $e$ is the nondeterministic symbol ``$*$''.)
Note that the abstraction formulae of (\ref{eq:trans1}) and (\ref{eq:trans2})
are computable. Indeed, for an array formula $\tau$ and a clause $C$,
the validity check of $(\tau \wedge a' = a\{t \leftarrow e\}) \Rightarrow C$
can be compiled into a decidable array theory supported by most SMT solvers
(cf.~\cite{kroening-book}).

\subsection{Correctness of the Abstraction Procedure}

Let $\phi := \delta x_1 \cdots \delta x_n.\,\varphi$ be a singly indexed array formula,
where $\varphi$ is quantifier-free.
Recall that, to compute a regular abstraction of $\phi$,
our procedure first computes an abstraction formula
$\phi^* = \delta x_1 \cdots \delta x_n.\,\Tilde{\varphi}$ by replacing
$\varphi$ with a quantifier-free formula $\Tilde{\varphi}$
such that
(1) $\Tilde{\varphi}$ is in negation normal form (NNF),
(2) $\Tilde{\varphi}$ is expressible in $\pred{P}$, and
(3) $\varphi \Rightarrow \Tilde{\varphi}$ holds in $\theory$.
Our procedure then outputs the formula $\phi^\pred{P} = \abs(\phi^*)$,
where $\abs(\cdot)$ replaces each atom $P_k[t/\ii]$ (resp.~$\neg P_k[t/\ii]$) in the formula
with {$\psi \Rightarrow \Delta_k(\sigma,w)$ (resp.~$\psi \Rightarrow \neg\Delta_k(t,w)$),
using $\psi$ to encode the boundary check of $P_k[t/\ii]$}.
Now we show that this procedure is sound, that is, $(\sem{\phi^\pred{P}}, \pred{P})$ is indeed a regular abstraction of $\phi$.

\begin{lemma}
  Let $\pred{P}$ be a set of indexed predicates, and $\phi$ be a singly indexed array formula
  that is expressible in $\pred{P}$. If $\phi$ is satisfiable, then $\abs(\mathsf{NNF}(\phi))$ is also satisfiable,
  where $\mathsf{NNF}(\phi)$ denotes the negation normal form of $\phi$. 
  \label{lemm:expressibel-sat}
\end{lemma}

\begin{proof}
Fix a set of predicates $\pred{P} = \iset{ P_1, \ldots, P_m }$.
For an SIA formula $\phi$ over $\voc$, suppose that $\sigma$ is a solution of $\phi$.
Then $\sigma$ uniquely determines a word $w = w(\sigma) \in \Sigma_m^*$ such that
$\sigma \models \Phi^{\pred{P}}_w$ and
$\size{w} = \max \{\sigma(x) : x$ is a parameter in $\voc\}$.
Particularly, for each $i \in \NatInt{1}{\size{w}}$ and $k \in \NatInt{1}{m}$,
$w \models \Delta_k(i,w)$ in $\struct_m$
if and only if
$\sigma \models P_k[i/\ii]$ in $\theory$ (cf.~Section~\ref{sec:indexed-predicate-abstraction}).
Now we argue that $(\sigma_{int},w)$ is a solution of $\abs(\phi)$,
where $\sigma_{int}$ denotes the restriction of $\sigma$ to index variables and parameters in $\voc$.
We shall prove this by induction on the quantifier rank of $\phi$, denoted by $rank(\phi)$.

For the base case, we show that $\sigma_{int}, w \models \abs(\phi)$ when $rank(\phi) = 0$ (i.e.~$\phi$ is quantifier-free).
Since $\phi$ is in NNF, we can assume w.l.o.g.~that $\phi = g\wedge h$ and $\abs(\phi) = g \wedge \abs(h)$ by writing $\phi$ in DNF.
Here, $g$ is a conjunction of index expressions, and $h = \bigwedge_j h_j$ is a conjunction of data expressions.
Assume to the contrary that 
$\sigma \models \phi$ but 
$\sigma_{int},w \not\models \abs(\phi)$.
Since $\sigma_{int} \models g$,
we have $\sigma_{int},w \models \neg\abs(h_j)$ for some $j$.
Since $h_j$ is expressible in $\pred{P}$,
suppose w.l.o.g.~that $h_j = P_k[t/\ii]$ for some $P_k \in \pred{P}$ and index term $t$.
Then $\sigma_{int},w \models 1 \le t \wedge t \le \size{w} \wedge \neg \Delta_k(t, w)$ holds by the definition of $\abs$.
But this implies $\sigma \models \neg h_j$ (and hence $\sigma \not\models \phi$) by the definition of $w$, leading to a contradiction.
It follows that $\sigma_{int},w \models \abs(\phi)$, i.e., the hypothesis is true when $rank(\phi) = 0$.

Now suppose that $rank(\phi)=n+1$, and that the hypothesis holds for formulae of rank equal to $n$.
If $\sigma \models \exists i.\,\psi$,
we have $\sigma \models \psi[p/i]$ for some $p$.
Then $\sigma_{int},w \models \abs(\psi[p/i])$ by the induction hypothesis,
which implies that $\sigma_{int},w \models \abs(\exists i.\,\psi)$.
Similarly, if $\sigma \models \forall i.\,\psi$,
we have $\sigma \models \psi[p/i]$ for all $p$.
Then $\sigma_{int},w \models \abs(\psi[p/i])$ for all $p$ by the induction hypothesis,
which implies that $\sigma_{int},w \models \abs(\forall i.\,\psi)$.
The statement therefore holds for formulae of all ranks by induction.
\end{proof}

\begin{theorem}[Correctness]
  \label{thm:correctness}
  Let $\pred{P}$ be a set of indexed predicates, and $\phi$ be a singly indexed array formula.
  Then %
  $\sem{\phi^\pred{P}}$ is regular and $\alpha({\phi}) \subseteq \sem{\phi^\pred{P}}$.
  Namely, $\phi^\pred{P}\!$ yields a regular abstraction of $\phi$.
\end{theorem}
\begin{proof}
$\sem{\phi^\pred{P}}$ is regular by Proposition~\ref{prop:logic-automata}.
Since $\alpha$ is monotonic and $\phi \Rightarrow \phi^*$,
it suffices to show that $\alpha({\phi^*}) \subseteq \sem{\phi^\pred{P}}$.
Consider an arbitrary abstract state $(\sigma,w) \in \alpha({\phi^*})$,
with $\sigma$ being an interpretation of index variables and system parameters $x_1,\ldots,x_r$ in $\voc$.
Define
$\Psi_{\sigma,w} := (\bigwedge_{i=1}^r x_i = \sigma(x_i)) \wedge \Phi_w^\pred{P}$.
Then $\phi^* \wedge \Psi_{\sigma,w}$ is satisfiable by the definition of $\alpha$.
Since $\phi^* \wedge \Psi_{\sigma,w}$ is in NNF and expressible in $\pred{P}$,
$\abs(\phi^* \wedge \Psi_{\sigma,w})$ is satisfiable by Lemma~\ref{lemm:expressibel-sat}.
However, observe that $(\sigma,w)$ is a solution of $\abs(\phi^* \wedge \Psi_{\sigma,w})$ when $\abs(\phi^* \wedge \Psi_{\sigma,w})$ is satisfiable,
as the formula only constrains its free word variable on the first $\size{w}$ letters.
Therefore,
we have $(\sigma,w)\in \sem{\abs(\phi^* \wedge \Psi_{\sigma,w})} = \sem{\abs(\phi^*) \wedge \abs(\Psi_{\sigma,w})}
\subseteq \sem{\abs(\phi^*)} = \sem{\phi^\pred{P}}$.
This allows us to conclude that $\alpha({\phi^*}) \subseteq \sem{\phi^\pred{P}}$,
and therefore $\alpha({\phi}) \subseteq \sem{\phi^\pred{P}}$.
\end{proof}

\subsection{A Closure Property}

Consider a finite set $\Phi$ of singly indexed array formulae.
Note that we can always select a set of indexed predicates that expresses all formulae in $\Phi$.
In such cases, if $\Phi$ is closed under Boolean operations, then
the abstractions our procedure computes for $\Phi$ are also closed under Boolean operations,
as is indicated by the following proposition.

\begin{proposition}
For indexed predicates $\pred{P}$
and singly indexed array formulae $\phi$ and $\psi$,
it holds that $(\phi \vee \psi)^\pred{P}\! \equiv \phi^\pred{P}\! \vee \psi^\pred{P}\!$.
Furthermore, we have $(\phi \wedge \psi)^\pred{P}\! \equiv \phi^\pred{P}\! \wedge \psi^\pred{P}\!$
and $(\neg \phi)^\pred{P}\! \equiv \neg \phi^\pred{P}\! \wedge \mathsf{true}^{\pred{P}}$\footnote{
One may regard $\sem{\mathsf{true}^{\pred{P}}}$ as the abstract state space induced by $\pred{P}$
in our regular abstraction framework.
}
when $\phi$ and $\psi$ are expressible in $\pred{P}$.
\label{prop:closure}
\end{proposition}
\begin{proof}
Suppose w.l.o.g.~that
$\phi = \delta x_1 \cdots \delta x_p.\bigvee_{i\in I} (g_i \wedge h_i)$ and 
$\psi = \delta y_1 \cdots \delta y_q.\bigvee_{j\in J} (g_j \wedge h_j)$,
where (i) $I \cap J = \emptyset$,
and (ii) for each $k \in I \uplus J$, $g_k$ is a conjunction of index expressions,
and $h_k$ is a conjunction of data expressions.
Then we have
\[
(\phi \vee \psi)^\pred{P} \equiv \delta x_1 \cdots \delta x_p\,\delta y_1 \cdots \delta y_q.
\bigvee\nolimits_{k \in\, I\,\uplus\,J} g_k \wedge \abs(\constr{h_k}) \equiv \phi^\pred{P}\! \vee \psi^\pred{P}\!.
\]
Furthermore, notice that
\begin{align*}
(\phi \wedge \psi)^\pred{P} & \equiv \delta x_1 \cdots \delta x_p\,\delta y_1 \cdots \delta y_q.
\bigvee\nolimits_{i \in I}\bigvee\nolimits_{j \in J} g_i \wedge g_j \wedge \abs(\constr{h_i \wedge h_j})\\
\phi^\pred{P}\! \wedge \psi^\pred{P} & \equiv 
\delta x_1 \cdots \delta x_p\,\delta y_1 \cdots \delta y_q.
\bigvee\nolimits_{i \in I}\bigvee\nolimits_{j \in J} g_i \wedge g_j \wedge \abs(\constr{h_i}) \wedge \abs(\constr{h_j})
\end{align*}
Hence, it suffices to show that
$\abs(\constr{h_i \wedge h_j}) \Leftrightarrow \abs(\constr{h_i}) \wedge \abs(\constr{h_j})$ 
holds in $\struct_m$.
For the ``$\Rightarrow$'' direction,
note that a feasible clause for $\constr{h_i}$ \emph{or} $\constr{h_j}$
is also feasible for $\constr{h_i \wedge h_j}$.
Thus, a solution of
$\abs(\constr{h_i \wedge h_j})$
is also a solution of
$\abs(\constr{h_i}) \wedge \phi(\constr{h_j})$.
For the ``$\Leftarrow$'' direction, suppose that $C$ is a feasible clause for
$\constr{h_i \wedge h_j}$, and thus
$\abs(\constr{h_i \wedge h_j}) \Rightarrow \abs(C)$ holds in $\struct_m$. 
Since $h_i$, $h_j$ are expressible in $\pred{P}$,
$\neg h_j \vee C$ and
$\neg h_i \vee C$ are feasible clauses for
$\constr{h_i}$ and $\constr{h_j}$, respectively.
Similarly, each atom of $h_i$ and $h_j$
induces a singleton feasible clause
for $\constr{h_i}$ and $\constr{h_j}$, respectively.
Hence,
$\abs(\constr{h_i} \wedge \constr{h_j}) \Rightarrow
\abs(h_i \wedge h_j \wedge (\neg h_i \vee C) \wedge (\neg h_j \vee C))
\Rightarrow \abs(h_i \wedge h_j \wedge C) \Rightarrow \abs(C)$ holds in $\struct_m$.
Thus, a solution of
$\abs(\constr{h_i}) \wedge \abs(\constr{h_j})$
is also a solution of
$\abs(\constr{h_i \wedge h_j})$.
This completes the proof for both directions, leading to $(\phi \wedge \psi)^\pred{P}\! \equiv \phi^\pred{P}\! \wedge \psi^\pred{P}\!$.

Finally, note that 
$(\neg \phi)^\pred{P}\! \wedge \phi^\pred{P}\! \equiv (\neg\phi \wedge \phi)^\pred{P}\! \equiv \mathsf{\false}^{\pred{P}}\! \equiv \mathsf{\false}$,
and $(\neg \phi)^\pred{P} \vee \phi^\pred{P}\! \equiv (\neg\phi \vee \phi)^\pred{P}\! \equiv \mathsf{true}^{\pred{P}}\!$.
This allows us to deduce that $(\neg \phi)^\pred{P}\! \equiv \neg\phi^\pred{P}\! \wedge \mathsf{true}^{\pred{P}}\!$.
\end{proof}

Proposition~\ref{prop:closure} makes it possible to compute precise regular abstractions of
complex formulae by composing abstractions of simple formulae.
For example, we may compute the regular abstraction of (\ref{eq:trans2})
as $\phi^\pred{P} = \psi^\pred{P} \wedge (\mathsf{copy\_except}(a) \wedge a' = a\{t \leftarrow e\})^\pred{P}$.
The two regular abstractions on the right-hand side can then be
reused for computing abstractions of other formulae.
Even if $\phi$ is not expressible in $\pred{P}$,
we can still compute $\phi^\pred{P}$ using composition
since regular abstractions are closed under conjunction (Proposition~\ref{prop:composition}),
but at the price of potentially losing precision.

We conclude this section with the following observation.

\begin{proposition}
Let $\pred{P}$ be a set of indexed predicates, and $\phi$ be a singly indexed array formula that is expressible in $\pred{P}$.
Then it holds that $\gamma(\phi^\pred{P}) = \gamma(\alpha(\phi))$.
\end{proposition}
\begin{proof}
Since $\alpha({\phi}) \subseteq \sem{\phi^\pred{P}}$ by Theorem~\ref{thm:correctness}, we have $\gamma(\alpha({\phi})) \subseteq \gamma({\phi^\pred{P}})$.
Suppose to the contrary that there exists $s \in \gamma({\phi^\pred{P}}) \setminus \gamma(\alpha({\phi}))$.
Since $s \in \gamma(\phi^\pred{P})$, we have $\alpha(s) \subseteq \sem{\phi^\pred{P}}$ by Proposition~\ref{prop:galois}.
Similarly, since $s \notin \gamma(\alpha(\phi))$, we have $\alpha(s) \not\subseteq \alpha(\phi)$ by Proposition~\ref{prop:galois}.
The latter fact implies that $s \not\models \phi$, that is, $s \models \neg \phi$.
Thus, $\alpha(s) \subseteq \alpha(\neg\phi) \subseteq \sem{(\neg\phi)^\pred{P}} = \sem{\true^\pred{P}} \setminus \sem{\phi^\pred{P}}$
by Theorem~\ref{thm:correctness} and Proposition~\ref{prop:closure}.
It follows that $\alpha(s) \cap \sem{\phi^\pred{P}} = \emptyset$, a contradiction. This completes the proof.
\end{proof}

%% file: verification.tex
\section{Verification of temporal array properties}
\label{sec:rmc}
In this section, we discuss how to verify linear-time array system properties
by combining our abstraction techniques with regular model checking.
We shall first introduce the notion of temporal array properties and abstractable specifications. 
We then present two verification methods for typical safety and liveness array properties.
Finally, we describe a generic technique to verify
a syntactic fragment of array properties called the index-bounded monodic properties.

\subsection{Temporal Array Property}

To express the temporal properties of an array system,
we provide a specification language combining the indexed array logic with LTL.
This can be seen as a restriction of FO-LTL \cite{hodkinson1999decidable,abadi1989power} to the indexed array logic.
For ease of exposition,
we shall only consider the ``globally'' and ``eventually'' connectives in the sequel.
Our approach however can be extended to handle
other standard connectives such as ``next'' and ``until''
in a straightforward manner.

A~\emph{temporal array property} is a formula $\phi$ constrained by
\[
    \phi ::= \psi \mid \neg \phi \mid \phi \vee \phi \mid \phi \wedge \phi \mid \exists i.\,\phi \mid \forall i.\,\phi \mid \G \phi \mid \F \phi
\]
where $\psi$ is an indexed array formula, $i$ is an index variable, and $\G\!$ and $\F\!$ are
the standard ``globally'' and ``eventually'' temporal connectives, respectively, in LTL.
As usual, we extend the logic with logical operators such as $\Rightarrow$ (implies) and $\Leftrightarrow$ (if and only if).
The semantics of a temporal array property is standard (see e.g.~\cite{hodkinson1999decidable}).
Intuitively, %
temporal connectives offer means of relating two states at different time points in a path, whereas the first-order quantifiers allow one to relate different array positions.
For example, both $\G (\forall i.\,1 \leq i \wedge i \leq |a| \Rightarrow a[i] = 0)$ and
$\forall i.\,(1 \leq i \wedge i \leq |a| \Rightarrow \G a[i] = 0)$
assert that the array $a$ has only 0 as elements throughout the path.
We say that a temporal array property $\phi$ holds for an array system if the system does not have a path that satisfies $\neg\phi$.

\subsection{Abstraction of Array System Specifications}

\begin{definition}[Abstractable specification]
An array formula is \emph{abstractable}
if it is of form $\exists \args{i}.\: \phi$
for some (possibly quantified) SIA formula $\phi$.
An array system
is \emph{abstractable} if it is specified with abstractable array formulae.
A temporal array property $\phi$ is \emph{abstractable}
if it is constrained by
\begin{align*}
    \phi & ::= \psi \mid \phi \vee \phi \mid \phi \wedge \phi \mid \exists i.\,\phi \mid \G \phi \mid \F \phi
\end{align*}
for (possibly quantified) SIA formulae $\psi$.
\end{definition}

Abstractable formulae slightly generalize SIA formulae
by allowing multiple existentially indexed/quantified variables. 
In fact, a formula is abstractable if and only if it can be transformed into an SIA formula using \emph{Skolemization}.
Thus, an abstractable first-order array system $\structT$
can be formally transformed into a set of SIA formulae.
More precisely, suppose that $\phi := \exists \args{i}.\, \psi$
is an abstractable array formula over variables $\voc$.
Define an SIA formula $\psi^\star := \psi [\args{c}/\args{i}]$
over variables $\voc \uplus \{\args{c}\}$, where $\args{c}$ are fresh system parameters.
We can then add the new parameters $\args{c}$ to the vocabulary of $\structT$
and substitute $\psi^\star$ for $\phi$ in the specification. %
Clearly, a temporal array property holds for $\structT$
if and only if the property holds for the Skolemized version of $\structT$.

An abstractable temporal array property can
furthermore be converted to an equisatisfiable temporal array property
wherein all maximal non-temporal subformulae are singly indexed.
The conversion is essentially the same as Skolemization,
except that here we replace an indexed variable quantified inside
a $\G\!$ operator with a fresh index variable, and
replace an indexed variable that is \emph{only} quantified inside
an $\F\!$ operator with a fresh system parameter.
The maximal non-temporal subformulae of the obtained temporal array property are all singly indexed,
meaning that their semantics can be overapproximated by regular abstractions.
In other words, an abstractable temporal array property can be abstracted
to a propositional temporal formula with atoms being regular languages.
This fact makes it possible to leverage
existing verification techniques in regular model checking
to perform abstract analysis for abstractable specifications and properties.

\subsection{Safety Verification}

A \emph{safety array property} is a temporal array property in form of
$\forall \args{i}.\, (\psi_1 \Rightarrow \G \psi_2)$.
Fix a safety array property $\phi$ and an
array system $\structT := (\voc, \phi_I, \phi_T)$.
Suppose that both $\structT$ and
$\neg \phi \equiv \exists \args{i}.\, (\psi_1 \wedge \F \neg \psi_2)$
are abstractable.
Define $\psi_1^\star := \psi_1[\args{c}/\args{i}]$ and
$\psi_2^\star := \psi_2[\args{c}/\args{i}]$
with fresh system parameters $\args{c}$.
Moreover, define a transition system
$\structT^\star := (\voc^\star, \phi_I^\star, \phi_T, \phi_B)$
with
$\voc^\star := \voc \uplus \{\args{c}\}$,
$\phi_I^\star := \phi_I \wedge \psi_1^\star$,
and $\phi_B := \neg \psi_2^\star$.
It is clear that $\structT^\star$ is abstractable.
Given a set $\pred{P}$ of indexed predicates over $\voc^\star$,
we can compute regular abstractions
$(I, \pred{P})$, $(T, \pred{P})$, and $(B, \pred{P})$
of $\phi_I^\star$, $\phi_T$, and $\phi_B$,
respectively, as described in Section~\ref{sec:computation-of-regular-abstractions}.
We then obtain a regular safety property $(I, T, B)$
that holds iff $T^*(I) \cap B = \emptyset$,
where $T^*$ denotes the transitive closure of $T$.
It is clear that $(I, T, B)$
holds only if $\structT$ satisfies the safety array property $\phi$.
With Proposition~\ref{prop:logic-automata}, we can summarize this result as follows.

\begin{theorem}
  \label{thm:regulation-abstraction-for-safety}
    Let $\structT$ be an array system,
    and $\phi$ be a safety array property.
    Suppose that $\structT$ and $\neg \phi$ are abstractable.
    Given a set of indexed predicates,
    we can effectively compute finite automata $\cI$, $\cT$, and $\cB$
    such that $\phi$ holds for $\structT$ if the safety property $(\cI, \cT, \cB)$ holds.
\end{theorem}

\subsection{Liveness Verification}
A \emph{liveness array property} is a temporal array property
in form of $\forall \args{i}.\, (\psi_1 \Rightarrow \F \psi_2)$.
(We pick this ``eventuality'' property for simplicity and sufficiency for our benchmarks, although our technique can easily be adapted to other liveness properties such as recurrence).
Fix a liveness array property $\phi$ and an
array transition system $\structT := (\voc, \phi_I, \phi_T)$.
Suppose that both $\structT$ and
$\neg \phi \equiv \exists \args{i}.\,(\psi_1 \wedge \G \neg \psi_2)$
are abstractable.
Define $\psi_1^\star := \psi_1[\args{c}/\args{i}]$
and $\psi_2^\star := \psi_2[\args{c}/\args{i}]$
with fresh system parameters $\args{c}$,
and an array transition system
$\structT^\star := (\voc^\star, \phi_I^\star, \phi_T, \phi_F)$,
where
$\voc^\star := \voc \uplus \{\args{c}\}$,
$\phi_I^\star := \phi_I \wedge \psi_1^\star$,
and $\phi_F := \psi_2^\star$.
It is clear that $\structT^\star$ is abstractable.
Given a set $\pred{P}$ of indexed predicates over $\voc^\star$,
we compute regular abstractions
$(I, \pred{P})$, $(T, \pred{P})$, and $(E, \pred{P})$ for
$\phi_I^\star$, $\phi_T$, and $\neg\phi_F$, respectively.
Now, following the reduction described in Section~\ref{sec:abstract-analysis},
we let $F := \Sigma_m^*\setminus E$ and define a regular liveness property $(I, T, F)$
that holds if and only if
$st(\pi) \cap F \neq \emptyset$ for every path $\pi \in \Pi(I,T)$.
It is not hard to see that the abstract transition system $\iset{\Sigma_m^*,I,T}$ preserves
\emph{infinite} counterexample paths from $\structT$. That is,
if there exists an infinite concrete path $s_0\, s_1 \cdots$
such that $s_0 \models \psi_I^\star$
and $s_i \not\models \phi_F$ for all $i \ge 0$,
then there exists an infinite abstract path $t_0\,t_1 \cdots$
such that $t_0 \in I$ and $t_i \notin F$ for all $i \ge 0$.

We say that a transition system is~\emph{$\neg\phi$-progressing}
if every maximal system path satisfying $\neg \phi$ is an infinite path.
Observe that when $\structT$ is $\neg\phi$-progressing,
$(I, T, F)$ holds only if $\structT$ satisfies the liveness array property $\phi$.
We hence have the following result.

\begin{theorem}
  \label{thm:regulation-abstraction-for-liveness}
  Let $\structT$ be an array system,
  and $\phi$ be a liveness array property.
  Suppose that $\structT$ and $\neg \phi$ are abstractable,
  and that $\structT$ is $\neg \phi$-progressing.
  Then given a set of indexed predicates,
  we can effectively compute finite automata $\cI$, $\cT$, and $\cF$
  such that $\phi$ holds for $\structT$
  if the liveness property $(\cI, \cT, \cF)$ holds.
\end{theorem}

\subsubsection*{Progress verification}
Theorem \ref{thm:regulation-abstraction-for-safety}
and \ref{thm:regulation-abstraction-for-liveness}
exploit the same technique that
overapproximates the target array property
with a counterpart property of the abstract system.
Generally, this technique works for any temporal array property
with an abstractable negated formula,
as long as a false property remains false
after new transitions are introduced to the system.
For a safety property, this requirement is met automatically
since extensions of a counterexample path are still counterexample paths.
For a liveness array property $\phi$,
this requires the array system to be $\neg\phi$-progressing,
i.e., the system does not have a finite maximal counterexample path.

We may formulate and verify the progress condition as a safety array property as follows.
Given an array system $\structT := (\voc, \phi_I, \phi_T)$
and a liveness array property
$\phi := \forall \args{i}.\, (\psi_1 \Rightarrow \F \psi_2)$,
we aim to check if $\structT$ is $\neg\phi$-progressing.
As before, define $\psi_1^\star := \psi_1[\args{c}/\args{i}]$
and $\psi_2^\star := \psi_2[\args{c}/\args{i}]$
with fresh system parameters $\args{c}$.
We then specify an array system
$\structT^* := (\voc^\star, \phi_I^\star, \phi_T^\star)$
with $\voc^\star := \voc \uplus \{\args{c}\}$,
$\phi_I^\star := \phi_I \wedge \psi_1^\star$, and
$\phi_T^\star := \phi_T \wedge \neg \psi_2^\star$.
It is easy to see that the system $\structT^*$ is safe
w.r.t.~the safety array property
$\phi^* := \forall \args{i}.\,\G \exists \args{i}'\!.\,\phi_T^\star$
if and only if $\structT$ is $\neg\phi$-progressing.
When the transition formula $\phi_T$
is specified with guarded commands,
say with $n$ commands $(A_1 \allows B_1), \ldots, (A_n \allows B_n)$,
the safety property $\phi^*$ can be written as
$\forall \args{i}.\,\G (A_1^\star \vee \cdots \vee A_n^\star \vee \psi_2^\star)$.
Notably, when both $\structT$ and $\neg\phi$
are abstractable (which is precisely the assumption of Theorem~\ref{thm:regulation-abstraction-for-liveness}),
the induced system $\structT^*$ and safety property $\phi^*$ are also abstractable.
Consequently, the progress condition of $\structT$ with respect to $\neg\phi$
can be checked formally using the safety verification method stated in
Theorem~\ref{thm:regulation-abstraction-for-safety}.

\subsection{Liveness Verification under Fairness Requirements}
\label{sec:fairness}
We briefly discuss how to perform liveness verification of a subclass of array systems in the presence of fairness requirements. Such requirements are essential for specifying reactive and concurrent process systems (cf.~\cite{demri2016temporal,manna2012temporal}).

\begin{definition}[Index-bounded array system]
    An array system $(\voc, \phi_I, \phi_T)$ is~\emph{index-bounded}
    if for each index variable $x \in \voc$, there exist two parameters $l_x, h_x \in \voc$
    such that $s(l_x) \le s(x) \le s(h_x)$ holds in any reachable state $s$ of the system.
    In other words, the valuation of $x$ is bounded between $l_x$ and $h_x$ during system executions.
\end{definition}

Intuitively, index-boundedness requires the range of index variables to be finite at runtime,
while leaving the range of array elements unrestricted.
This assumption is reasonable since in practice, index variables are mostly used to access elements of a finite array.
Index-boundedness can be specified formally in the array system specification,
or checked against a given array system by formulating the assumption as a safety array property.

\begin{definition}[Fairness specification]
    An~\emph{(abstractable) fairness specification}
    is a temporal array property $\lambda := \forall i.\,\eta$ constrained by the grammar
    \begin{align*}
        \eta & ::= \GF \psi \mid \FG \psi \mid \eta \vee \eta \mid \eta \wedge \eta,
    \end{align*}
    where $\psi$ is an abstractable indexed array formula.
    We say that a path $\pi$ is \emph{fair with respect to~$\lambda$},
    or simply \emph{fair} when $\lambda$ is clear,
    if $\pi$ satisfies $\lambda$.
    A temporal array property $\phi$ holds for an array system under a fairness specification
    if the system does not have a fair path satisfying $\neg\phi$.
\end{definition}

We note that our formalism of fairness is expressive enough to capture typical fairness
requirements for reactive and concurrent systems, including
\emph{process fairness} (in form of $\GF \phi$),
\emph{weak fairness} (in form of $\GF \phi \vee \GF \psi$),
and \emph{strong fairness} (in form of $\GF \phi \vee \FG \psi$).
See e.g.~\cite{manna2012temporal} for a detailed discussion of these fairness requirements.

Given fairness specification $\lambda$ and indexed predicates $\pred{P}$,
we can compute a regular abstraction $(E,\pred{P})$
overapproximating $\lambda$ in the abstract domain.
Namely, if there is a fair concrete path $s_0\, s_1 \cdots$ with respect to $\lambda$,
then there is a fair abstract path $t_0\, t_1 \cdots$ with respect to $E$,
and $t_i \in \alpha(s_i)$ for each $i \ge 0$.
When the abstract system is \emph{weakly finite}~\cite{esparza2012proving},
that is, every infinite path eventually enters a cycle,
checking the existence of a fair counterexample path essentially amounts to
searching for a reachable fair cycle,
which can in turn be reduced to a safety verification problem~\cite{schuppan2006liveness,daniel2016infinite}.
Indeed, for weakly finite systems, we can formally
translate the abstract fairness condition $(E,\pred{P})$
to a formalism similar to the so-called \emph{B\"uchi regular transition system (BRTS)}~\cite{abdulla2012regular},
which generalizes the translation from linear temporal properties to
B\"uchi automata in finite-state model checking~\cite{vardi1986automata}.
We refer the interested reader to~\cite{hong2022symbolic} for the technical details of our translation procedure.

It is easy to see that regular abstractions of an index-bounded array system
are weakly finite. Thus, together with Theorem~\ref{thm:regulation-abstraction-for-liveness},
we can summarize our results in this section as follows.
\begin{theorem}
 \label{thm:regulation-abstraction-for-liveness-2}
  Let $\structT$ be an index-bounded array system, $\phi$ be a liveness array property, and $\lambda$ be a fairness specification.
  Suppose that $\structT$ and $\neg \phi$ are abstractable,
  and that $\structT$ is $\neg \phi$-progressing.
  Then we can compute finite automata $\cI$, $\cT$, and $\cF$,
  as well as a regular fairness requirement $\Lambda$,
  such that $\phi$ holds for $\structT$ under the fairness specification $\Lambda$ if the liveness property $(\cI, \cT, \cF)$ holds under the fairness requirement $\Lambda$.
\end{theorem}

\subsection{Verification of Monodic Temporal Properties}
\label{sec:fair-reduction-for-monodic-formulae}

In this section, we briefly discuss a generic method to verify
a monodic fragment of the temporal array properties. %
This method directly extends the classical tableaux-theoretic
approach for LTL model checking (cf.~\cite{clarke2018model}) to infinite structures.
Similar techniques were employed in~\cite{abdulla2012regular} for verifying LTL(MSO) properties
of parameterized systems, and in~\cite{padon2021temporal} for verifying FO-LTL properties
of first-order transition systems.
The presentation here is essentially adopted from \cite{padon2021temporal}.

Consider a temporal array property $\phi$ over variables $\voc$.
We may eliminate $\F\!$ from $\neg\phi$
using the rule $\F\phi \equiv \neg \G \neg \phi$,
and assume that $\G\!$ is the only temporal connective contained in $\neg\phi$.
Let $sub(\neg\phi)$ denote the set of subformulae of $\neg\phi$. Define
\[
    \voc_{\neg\phi} := \voc \uplus \{\, g_{\varphi} : \G \varphi \in sub(\neg\phi) \},
\]
where each $g_{\varphi}$ is a fresh array variable.
In a nutshell, to check that an array system satisfies $\phi$,
we take the product of the system and a ``monitor'' of $\neg \phi$
over the extended set of variables $\voc_{\neg\phi}$.
The product system will be associated with a fairness condition
such that the projection of a \emph{fair} path from $\voc_{\neg\phi}$ to $\voc$
coincides with a path satisfying $\neg\phi$ in the original array system.
Thus, the temporal property $\phi$ holds for the original system
if and only if the product system has no fair path.

To apply the aforementioned reduction in our abstraction framework, however,
the temporal property $\phi$ needs to meet some technical conditions.
We say that $\phi$ is \emph{monodic}~\cite{hodkinson1999decidable}
if every \emph{temporal} subformula of $\phi$ contains at most one free variable outside $\voc$.
For instance, when $\voc = \{a,k\}$, the property $\forall j.\G a[j]>a[k]$ is monodic,
whilst $\exists i.\forall j.\G a[j]>a[i]$ is not.
We say that $\phi$ is \emph{index-bounded}
\cite{cimatti2022verification,emerson2003reasoning,abdulla2012regular,abdulla2016parameterized}
if there exists an index variable $h \in \voc$
such that each quantified variable of $\phi$ is bounded above by $h$.
This index-boundedness restriction may be imposed syntactically,
for example, by replacing each subformula $\exists i.\psi$
of $\phi$ with $\exists i.\,i \le h \wedge \psi$,
and each subformula $\forall i.\psi$ with
$\forall i.\,i \le h \Rightarrow \psi$.

For a monodic formula $\phi$ over $\voc$, we write $\phi(i)$
to indicate that $i$ is the only free variable of $\phi$ outside $\voc$,
if there is any. In the case that $\phi$ has no such free variable,
we just set $i$ to be an arbitrary index variable outside $\voc$.
We use $\op{FO}[\phi]$ to denote a first-order representation of $\phi$,
defined inductively as follows with $\delta \in \{ \exists, \forall\}$:
\begin{align*}
    \op{FO}[\phi] & := \phi\quad\mbox{($\phi$ is non-temporal)} & \op{FO}[\delta i.\,\phi] & := \delta i. \op{FO}[\phi]\\
    \op{FO}[\G \phi(i)] & := g_{\phi}[i] \neq 0 & \op{FO}[\phi_1 \vee \phi_2] & := \op{FO}[\phi_1] \vee \op{FO}[\phi_2]\\
    \op{FO}[\neg \phi] & := \neg \op{FO}[\phi]  & \op{FO}[\phi_1 \wedge \phi_2] & := \op{FO}[\phi_1] \wedge \op{FO}[\phi_2]
\end{align*}
Now, consider an index-bounded monodic property $\phi$ and an array system $\structT := (\voc, \phi_I, \phi_T)$.
Define a fair transition system $\structT_{f} := (\voc_{\neg\phi}, \psi_I, \psi_T)$
with fairness requirement $\lambda$, where
\begin{align*}
    \psi_I & := \phi_I \wedge \op{FO}[\neg \phi] \wedge \bigwedge\nolimits_{\G\!\varphi(i) \in sub(\neg \phi)} \size{g_{\varphi}} = h\\
    \psi_T & := \phi_T \wedge \forall i.\bigwedge\nolimits_{\G\!\varphi(i) \in sub(\neg \phi)}
        \left(i > h \vee (g_{\varphi}[i] \neq 0 \Leftrightarrow (\op{FO}[\varphi(i)] \wedge g_{\varphi}'[i] \neq 0))\right) \\
      \lambda & := \forall i.\bigwedge\nolimits_{\G\!\varphi(i) \in sub(\neg \phi)} \GF (i > h \vee g_{\varphi}[i] \neq 0 \vee \neg\op{FO}[\varphi(i)])
\end{align*}

The fair transition system $\structT_{f}$ is 
constructed as the product of the original system $\structT$
and a monitor of $\neg \phi$ over the extended variables $\voc_{\neg\phi}$.
Given a subformula $\G \varphi$ of $\neg \phi$,
the monitor updates the array $g_{\varphi}$ along a path of $\structT$
in accordance with whether or not $\G \varphi$ is satisfied by the current path.
More concretely, consider a path $\pi$ of $\structT$ and some $t \in \Nat$.
If $\pi \models \G \varphi(t)$, then the monitor
maintains the array value $g_{\varphi}[t]$ along the path $\pi$
to make sure $\varphi(t)$ is satisfied at every state of the path.
Otherwise, if $\pi \not\models \G \varphi(t)$,
then the monitor maintains the array value $g_{\varphi}[t]$
along the path $\pi$ to make sure $\varphi(t)$ is falsified at some state of the path.
In the second case, the fairness condition $\lambda$ guarantees
that the event of $\varphi(t)$ being falsified
will not be postponed forever.

For a fair path of $\structT_{f}$, %
the projection of the path to the original variables $\voc$ yields
a path of the original system $\structT$ satisfying $\neg \phi$.
Conversely, a path satisfying $\neg \phi$ in the original system $\structT$
can be lifted to the extended variables $\voc_{\neg\phi}$ to obtain a fair path of $\structT_{f}$.
As a consequence, one can verify the temporal array property $\phi$ for $\structT$
by checking fair termination of $\structT_{f}$. %

\begin{proposition}[\cite{abdulla2012regular,padon2021temporal}]
    Let $\phi$ be an index-bounded monodic array property, $\structT$ be an array system,
    and $\structT_{f}$ be the fair transition system induced by $\structT$ and $\phi$.
    Then $\phi$ holds for $\structT$ if and only if $\structT_{f}$ has no fair path.
\end{proposition}

Furthermore, when $\structT$ is $\neg\phi$-progressing and
the fair transition system $\structT_{f}$ is abstractable,
termination analysis of $\structT_{f}$ is amenable
to our liveness-to-safety reduction techniques in Section~\ref{sec:fairness}.
This fact, in tandem with Theorem~\ref{thm:regulation-abstraction-for-liveness-2}, leads to the following result:

\begin{theorem}
    Let $\phi$ be an index-bounded monodic array property,
    and $\structT$ be a $\neg\phi$-progressing index-bounded array system.
    If the fair transition system $\structT_{f}$ is abstractable,
    then we can compute finite automata $\cI$ and $\cT$
    as well as a fairness requirement $\Lambda$,
    such that $\phi$ holds for $\structT$ if the transition system
    $(\cI, \cT)$ always terminates under the fairness requirement $\Lambda$.
\end{theorem}

%% file: case-studies.tex
\section{Case Studies}
\label{sec:cases}
We present four case studies in this section:
two array programs implementing selection sort and a simplified version of merge sort,
and two distributed algorithms
Dijkstra's self-stabilizing protocol
and Chang-Roberts leader election protocol.
To make our presentation succinct,
we shall often omit the index range constraint in a formula
when the range is clear from the context.
For example, we shall write $\forall i.\: a[i] = 0$ instead of
$\forall i.\, 1 \le i \wedge i \le \size{a} \Rightarrow a[i] = 0$.

\subsection{Selection Sort and Merge Sort}
\label{sec:selection-sort}
For selection sort, we consider a selection sort program as follows:

\begin{adjustwidth*}{-1em}{}
\resizebox{.98\linewidth}{!}{
\begin{minipage}{1.1\linewidth}
\begin{flalign*}
  & pc = 0 \allows low := 1,\: high := \size{a},\: pc := 1\\
  & pc = 1 \wedge low \ge high \wedge 1 < high \allows high := high - 1,\: low := 1\\
  & pc = 1 \wedge low < high \wedge a[low] \le a[high] \allows low := low + 1\\
  & pc = 1 \wedge low < high \wedge a[low] > a[high] \allows a[high] := a[low],\: a[low] := a[high],\: low := low + 1
\end{flalign*}
\end{minipage}}
\end{adjustwidth*}
\vspace{.8em}

\noindent The array transition system $\structT := \TranSystem$ induced by this program is
\[
     \voc := \{ pc, low, high, a, \size{a} \},\quad
     \phi_I := pc = 0,\quad
     \phi_T := \phi_1 \vee \phi_2 \vee \phi_3 \vee \phi_4,
\]
where $\phi_1,\phi_2,\phi_3,\phi_4$ are transition formulae derived from the four guarded commands.
We verify the typical correctness properties of a sorting algorithm as follows.

\begin{longtable}[]{@{}ll@{}ll@{}}
    \toprule
    Property & Specification & Explanation \tabularnewline
    \midrule
    \endhead
    $P_1$ & $\G \forall i.\forall j.\,(i<j \wedge high<j \Rightarrow a[i] \le a[j])$ & Sortedness \tabularnewline
    $P_2$ & $(\forall i.\: a[i] = a_0[i]) \Rightarrow \G ((\forall i.\exists j.\:a[i] = a_0[j]) \wedge (\forall i.\exists j.\:a_0[i] = a[j]))\quad $ & Permutation \tabularnewline
    $P_3$ & $\F high \le 1$ & Termination \tabularnewline
    \bottomrule
\end{longtable}
\noindent Here,
$P_1$ is a safety property stating that at any step,
the array segment after the pivot position $high$ is sorted;
$P_2$ is a safety property stating that the array produced by the program
has the same content as the input array modulo multiplicities;
$P_3$ is a liveness property stating that the program eventually terminates.
These three properties together establish the total correctness of the program
for arrays with distinct values.

As for merge sort, we consider a nondeterministic merge sort program as follows: %

\begin{adjustwidth*}{0.4em}{}
\resizebox{.98\linewidth}{!}{
\begin{minipage}{1.1\linewidth}
\begin{flalign*}
  & pc = 0 \wedge sorted(low,mid) \wedge sorted(mid,high) \wedge \neg sorted(low,high) \wedge high \ge high_1 \allows low_1 := low,\: pc := 1 \\
  & pc = 1 \wedge low_1 < mid \wedge mid < high \wedge a[low_1] \le a[mid] \allows low_1 := low_1 + 1 \\
  & pc = 1 \wedge low_1 < mid \wedge mid < high \wedge a[low_1] > a[mid] \allows a[n] = a[low_1],\: ptr := low_1,\: pc := 2 \\
  & pc = 1 \wedge \neg (low_1 < mid \wedge mid < high) \allows high_1 := high,\: low := *,\: mid := *,\: high := *,\: pc := 0 \\
  & pc = 2 \wedge ptr < mid \allows a[ptr+1] := a[n],\: a[n] := a[ptr+1],\: ptr := ptr + 1\\
  & pc = 2 \wedge ptr \ge mid \allows a[low_1] := a[n],\: mid := mid + 1,\: low_1 := low_1 + 1,\: pc := 1
\end{flalign*}
\end{minipage}}
\end{adjustwidth*}
\vspace{.8em}

\noindent Here, $sorted(l,h) := \forall i.~l \le i \wedge i < h-1 \Rightarrow a[i]\le a[i+1]$
is an auxiliary formula expressing that the array segment $a[l\ldots h)$ is sorted.
This program starts with $pc=0$ and $high_1=1$, and guesses the values of $low$, $mid$, and $high$ within the range $\NatInt{1}{n}$.
If the program spots a merge opportunity by this guess (i.e.~the guard of the first command is satisfied),
it proceeds to $pc=1$ and performs an in-place array merge.
After the merge, the program returns to $pc=0$ and guesses the values of $low$, $mid$, and $high$ again.
We verify the following eventuality property of the program:
\[
    P_1 := \F (pc=0 \wedge (\neg sorted(low,mid) \vee \neg sorted(mid,high) \vee sorted(low,high) \vee high < high_1)),
\]
which states that the program eventually fails to pinpoint a merge opportunity at $pc=0$.
When $P_1$ holds, any execution run that consistently spots a merge opportunity at $pc=0$
eventually reaches $pc=0$ with no further such opportunities, at which point the array segment $a[1\ldots n)$ is sorted.
In other words, $P_1$ indicates that the program eventually produces a sorted array on the ``proper'' execution runs,
namely, the runs where the program consistently attempts to merge partially sorted array segments.

\subsection{Dijkstra's Self-Stabilizing Algorithm}
\label{case:dijkstra-algorithm}

Dijkstra's self-stabilizing algorithm~\cite{dijkstra1982self},
as introduced in Section~\ref{sec:illustrative-example},
can be specified in our modeling language as follows:

\resizebox{\linewidth}{!}{
\begin{minipage}{1.1\linewidth}
\begin{flalign*}
& pid = 1 \wedge a[pid] = a[{n}] \wedge a[pid] < {k}-1 \allows a[pid] := a[pid] + 1,\: pid := *\\
& pid = 1 \wedge a[pid] = a[{n}] \wedge a[pid] = {k}-1 \allows a[pid] := 0,\: pid := *\\
& pid > 1 \wedge pid \le n \wedge a[pid] \neq a[pid-1] \allows a[pid] :=a[pid-1],\: pid := *
\end{flalign*}
\end{minipage}}
\vspace{.3em}

\noindent Here, we have identified $\size{a}$ with ${n}$ and used $pid$ to denote the ID of the process to be selected by the scheduler.
For simplicity, we assume that only privileged processes can be scheduled. %
Dijkstra's algorithm is progressing under this scheduling assumption,
since there always exists at least one privileged process in the system.
(Recall that we have proved this fact in Example~\ref{ex:state-formula-abstraction-2}.)

Dijkstra's algorithm can be initialized with arbitrarily many privileged processes.
When $n \le k$, the system will converge to a stable state containing exactly one privileged process.
For convenience, define an auxiliary formula
    $priv(i) :=  (i = 1 \wedge a[i] = a[{n}]) \vee (i \neq 1 \wedge a[i] \neq a[i-1])$.
We can then express the self-stabilizing property of Dijkstra's algorithm as
\[
    P := \FG (\exists i.\: priv(i) \wedge \forall j.\: j \neq i \Rightarrow \neg priv(j)),
\]
meaning that the system eventually converges to exactly one privileged process and remains so forever.
To prove the self-stabilizing property $P$,
we create subgoals $P_1, \ldots, P_4$ for the property
as listed in the following table.
    \begin{longtable}[]{@{}ll@{}ll@{}}
    \toprule
    Property & Specification & Explanation \tabularnewline
    \midrule
    \endhead
    $Q_1$ & $a[1] = a[{n}] \wedge pid=1$ & Process 1 is privileged and scheduled \tabularnewline
    $Q_2$ & $\forall i.\: i \neq 1 \Rightarrow a[i] \neq a[1]$ & $x_i \neq x_1$ for all $i\neq 1$ \tabularnewline
    $Q_3$ & $\forall i.\: i \neq 1 \Rightarrow a[i] = a[i-1]$ & Process $i$ is unprivileged for all $i \neq 1$ \tabularnewline
    $Q_4$ & $\exists i.\:priv(i) \wedge \forall j.\: j \neq i \Rightarrow \neg priv(j)\qquad$ & Exactly one process is privileged \tabularnewline
    \midrule
    $P_1$ & $\GF Q_1$ & Recurrence property \tabularnewline
    $P_2$ & $\GF Q_1 \Rightarrow \F Q_2$ & Liveness property under fairness condition\tabularnewline
    $P_3$ & $Q_2 \Rightarrow \F (Q_1 \wedge Q_3)$ & Liveness property \tabularnewline
    $P_4$ & $(Q_1 \wedge Q_3) \Rightarrow \G\, Q_4$ & Safety property \tabularnewline
    $P$ & $\FG\, Q_4$ & Self-stabilizing property \tabularnewline
    \bottomrule
    \end{longtable}
Intuitively,
$P_1$ states that process $p_1$ is privileged and scheduled infinitely often;
$P_2$ states that if process $p_1$ is privileged and scheduled infinitely often,
then eventually variable $x_1$ differs from all the other variables;
$P_3$ states that if $x_1$ differs from all the other variables at some point,
then eventually $p_1$ is the only privileged process;
$P_4$ states that if $p_1$ is the only privileged process, then the system has stabilized.
To verify that Dijkstra's algorithm is self-stabilizing,
it suffices to check the properties $P_1, \dots, P_4$ separately,
since the self-stabilizing property
$P$ is subsumed by $P_1 \wedge P_2 \wedge P_3 \wedge P_4$.

\subsection{Chang-Roberts Algorithm}

The Chang-Roberts algorithm~\cite{chang1979improved}
is a ring-based leader election protocol.
The algorithm assumes that each process has a unique ID,
and that messages can be passed on the ring in the clockwise direction.
At first, all processes are active.
An active process becomes passive after it emits or forwards a message.
The algorithm starts when an active process turns into an initiator
and emits a message tagged with its ID to the next process.
Let $id_i$ denote the ID of Process $i$.
When Process $i$ receives a message tagged with $id$,
it reacts in three cases:
\begin{itemize}
\item
  $id < id_i$ : Process $i$ will purge the message.
\item
  $id > id_i$ : Process $i$ will forward the message and become passive.
\item
  $id = id_i$ : Process $i$ will announce itself as a leader.
\end{itemize}
The rationale behind the protocol is that only the message tagged with
the largest ID will complete the round trip and make its sender the leader.
We specify the protocol as follows:

\begin{adjustwidth*}{-0em}{}
\resizebox{\linewidth}{!}{
\begin{minipage}{1.2\linewidth}
\begin{align*}
& id[pid] < msg[pid+1] \wedge st[pid] = 0 \allows st[pid] := 1,\: pid := * \\
& id[pid] \ge msg[pid+1] \wedge st[pid] = 0 \allows st[pid] := 1,\, msg[pid+1] := id[pid],\, pid := * \\
& msg[pid] \ge id[pid] \wedge msg[pid] < msg[pid+1] \allows msg[pid] := 0,\, st[pid] := 1,\, pid := *\\
& msg[pid] \ge id[pid] \wedge msg[pid] \ge msg[pid+1] \allows msg[pid] := 0,\, st[pid] := 1,\, msg[pid+1] := msg[pid],\, pid := *
\end{align*}
\end{minipage}}
\end{adjustwidth*}
\vspace{.8em}

\noindent These four commands correspond to four operations as follows:
an initiator emits its ID, which is purged by the successor process;
an initiator emits its ID, which is cached by the successor;
a process forwards a message, which is purged by the successor;
a process forwards a message, which is cached by the successor.
As before, we use $pid$ to denote the ID of the scheduled process.
We use $st[i]$ to represent the status of Process $i$, with 0 and 1 standing for active and passive, respectively.
We use $id[i]$ and $msg[i]$ to represent the ID and the message buffer, respectively, of Process $i$.
We stipulate that all process IDs are positive, and $msg[i]=0$ if and only if the message buffer of Process $i$ is empty.
When a process receives multiple messages, it keeps the one attached with the largest ID.
Finally, we modify these commands to distinguish the case $pid = n$:
for this case, we replace $pid+1$ with 1.
We consider two correctness properties of the Chang-Roberts algorithm:
\begin{longtable}[]{@{}ll@{}ll@{}}
\toprule
Formula\quad & Definition & Description \tabularnewline
\midrule
\endhead
$largest(i)$\hspace{1em} & $\forall j.\: id[i] \ge id[j]$ & Process $i$ has the largest ID \tabularnewline
$unique(i)$    & $\forall j.\: j \neq i \Rightarrow id[i] \neq id[j]$\hspace{1.6em} & Process $i$ has a unique ID  \tabularnewline
$elected(i)$  & $msg[i] = id[i]$ & Process $i$ is elected as a leader \tabularnewline
\bottomrule
\end{longtable}
\begin{longtable}[]{@{}ll@{}}
\toprule
Property & Specification \tabularnewline
\midrule
\endhead
$P_1$ &
$\forall i.\, \neg largest(i) \wedge unique(i) \Rightarrow \G \neg elected(i)$  \tabularnewline
$P_2$ &
$(\forall i.\, \GF pid=i) \Rightarrow (\forall i.\, largest(i) \wedge unique(i) \Rightarrow \F elected(i))$ \tabularnewline
\bottomrule
\end{longtable}
\noindent
Here,
$P_1$ is a safety property saying that a process not holding the largest ID is never elected as a leader;
$P_2$ is a liveness property saying that a process holding the largest ID is eventually elected as a leader under the fairness requirement
that each process is scheduled infinitely often.
The validity of these properties relies on the assumption that the process IDs are unique.
Thus, we need to specify this assumption in the properties.

%% file: evaluation.tex
\section{Optimization and evaluation}
\label{sec:evaluation}
\begin{table*}[tb]
\renewcommand\thetable{2}
\protect\caption{Results of applying the $L^*$ learning-based model checker
    \cite{chen2017learning} and the SLRP tool~\cite{lin2016liveness}
    in the abstract analysis of the case studies.
    In the table, $P_0$ denotes the progress conditions,
    and $P_1,\ldots,P_4$ are the temporal properties defined in Section~\ref{sec:cases}.
    The table also shows the indexed predicates we used to compute regular abstractions.
    Note that these predicates contain Skolem constants 
    produced in the formulae rewriting step of our method (see Section~\ref{sec:rmc}).
    For each property, we present the total runtime (Total), the computation time of
    the initial and refined abstract systems (Model), and the time consumed by the model checker (Solver).
    }
\label{tab:array-systems-safety} 
\begin{centering}
\scalebox{0.75}{
\begin{tabular}{|l|c||c|c|c|}
\hline 
\multicolumn{2}{|c||}{\textbf{Safety Properties}} & \multicolumn{3}{c|}{$L^{*}$ }\tabularnewline
\hline 
\textsf{Name}  & \textsf{Indexed Predicates}  & \textsf{Total}  & \textsf{Model} & \textsf{Solver}\tabularnewline
\hline 
\hline 
s.sort $P_{0}$  & $a[\ii]\le a[high],\,a[\ii]\le a[p]$  & 0.5s  & 0.0s & 0.4s\tabularnewline
\hline 
s.sort $P_{1}$  & $a[\ii]\le a[high],\,a[\ii]\le a[p]$  & 3.7s  & 0.1s & 3.4s\tabularnewline
\hline 
s.sort $P_{2}$  & $a[\ii]=a_{0}[p],\,a_{0}[\ii]=a[p]$  & 6.3s  & 0.3s & 5.7s\tabularnewline
\hline 
m.sort $P_{0}$  & $a[\ii]\ge a[n],\,a[\ii]\le a[\ii+1],\,a[\ii]\le a[mid]$  & 0.9s  & 0.1s & 0.8s\tabularnewline
\hline 
Dijk. $P_{4}$  & $a[\ii]=a[\ii-1],\,a[\ii]=a[{n}],\,a[\ii]=a[1]$  & 1.2s  & 0.1s & 1.0s\tabularnewline
\hline 
C.-R. $P_{1}$  & $id[\ii]=id[p],\,id[\ii]<id[q],\,msg[\ii]=id[p],\,msg[\ii]<id[q]$  & 3.3s  & 0.4s & 2.7s\tabularnewline
\hline 
C.-R. $P_{0}$  & $\,msg[\ii]\neq0,\,st[\ii]=0,\,id[\ii]<id[p],\,msg[\ii]<id[p],\,msg[\ii]=id[p]\,$  & 8.4s  & 0.9s & 7.2s\tabularnewline
\hline 
\end{tabular}} 

\vspace{0.5em}
 \scalebox{0.73}{ %
\begin{tabular}{|l|c||c|c|c||c|c|c|}
\hline 
\multicolumn{2}{|c||}{\textbf{Liveness Properties}} & \multicolumn{3}{c||}{$L^{*}$ } & \multicolumn{3}{c|}{SLRP}\tabularnewline
\hline 
\textsf{Name}  & \textsf{Indexed Predicates}  & \textsf{Total}  & \textsf{Model} & \textsf{Solver} & \textsf{Total}  & \textsf{Model} & \textsf{Solver}\tabularnewline
\hline 
\hline 
s.sort $P_{3}$  & no predicates  & 0.7s  & 0.4s & 0.2s & 1.5s & 0.1s & 1.3s\tabularnewline
\hline 
m.sort $P_{1}$  & $a[\ii]\ge a[n],\,a[\ii]\le a[\ii+1],\,a[\ii]\le a[mid]$  & 5m47s & 5m39s & 5s & 1.4s & 0.1s & 1.1s\tabularnewline
\hline 
Dijk. $P_{1}$  & $a[\ii]=a[\ii-1],\,a[\ii]=a[{n}],\,a[\ii]=a[1]$  & t.o. & -- & -- & 12m54s & 3s & 12m50s\tabularnewline
\hline 
Dijk. $P_{3}$  & $a[\ii]=a[\ii-1],\,a[\ii]=a[{n}],\,a[\ii]=a[1]$  & t.o. & -- & -- & 3m14s & 2s & 3m10s\tabularnewline
\hline 
Dijk. $P_{2}$  & $a[\ii]=z,\,a[\ii]=a[1]$  & 9m13s  & 4s & 9m06s & n/a & -- & -- \tabularnewline
\hline 
C.-R. $P_{2}$  & $~st[\ii]=0,\,id[\ii]<id[p],\,msg[\ii]<id[p],\,msg[\ii]=id[p]~$  & 3m04s  & 11s & 2m51s & n/a & -- & -- \tabularnewline
\hline 
\end{tabular}}
    \end{centering}
\end{table*}
\begin{figure}[ht!]
   {\centering
    \scalebox{.85}{
    \begin{mathpar}
        \inferrule* %
        { a' = a\{I \leftarrow a[J]\} \quad a[J] \bowtie a[K] \quad I \neq K }
        { a'[I] \bowtie a'[K] }
        \hspace{1em}
        \inferrule* %
        { a' = a\{J \leftarrow e\} \quad a[I] \bowtie a[K] \quad J\neq I \quad J \neq K}
        { a'[I] \bowtie a'[K] }
        \and
        \inferrule* %
        { x' = J \quad a[I] \bowtie a[K] }
        { a'[I] \bowtie a'[K] }
        \hspace{1em}
        \inferrule* %
        { a' = a\{I \leftarrow a[J]\} \quad a[J] \bowtie n }
        { a'[I] \bowtie n }
        \hspace{1em}
        \inferrule* %
        { a' = a\{J \leftarrow e\} \quad a[I] \bowtie n \quad I\neq J }
        { a'[I] \bowtie n }
        \and
        \inferrule* %
        { x' = J \quad a[I] \bowtie b[K] }
        { a'[I] \bowtie b'[K] }
        \hspace{.85em}
        \inferrule* %
        { a' = a\{I \leftarrow b[J]\} \quad b[J] \bowtie c[K]}
        { a'[I] \bowtie c'[K] }
        \hspace{.85em}
        \inferrule*
        { a' = a\{J \leftarrow e\} \quad a[I] \bowtie b[K] \quad I \neq J }
        { a'[I] \bowtie b'[K] }
    \end{mathpar}}}
    \caption{Example constraint templates for updating and copying array contents, where $\bowtie\; \in \{=, \neq, <, \le, >, \ge\}$.
    Each template consists of multiple \emph{premises} and a \emph{consequence}.
    The leftmost premise corresponds to an update (i.e.~the assignment part of a guarded command),
    and the rest of the premises come from the guard. The consequence copies a predicate value from the current state to the next state.
    These templates exploit a simplifying assumption that one guarded command updates precisely one index variable or array element.}
    \label{fig:constraint-templates}
\end{figure}

We have implemented a prototype to evaluate our approach over the case studies.
The verification of an array system consists of two stages:
(i) computing regular abstractions from the system specification,
and (ii) performing abstract analysis based on these regular abstractions.
When we compute a regular abstraction as described in Section~\ref{sec:abstraction-procedure},
a crucial step is to compute the constraint $\constr{\psi}$ at (\ref{def:abstraction-constraint}),
which could be obtained by enumerating the minimal unsatisfiable cores (MUCs).
For this purpose, we first compute MUCs for $\psi$ using the MUST tool~\cite{bendik2020must}.
Since the number of MUCs could be large,
including all clauses induced by these MUCs in $\constr{\psi}$ is generally impractical.
Instead, our tool computes the constraint in an incremental manner:
we sort these clauses by size, and include a clause
only when the constraints generated by the smaller clauses
lead to spurious counterexample or timeout in the abstract analysis.
As our experimental results have shown,
most of the properties in our case studies can be verified by including
a moderate number of clauses in the state formulae abstractions.

For transition formulae abstractions, %
we additionally use templates to capture constraints involving array updates (Figure~\ref{fig:constraint-templates}).
We syntactically populate these templates with predicates and index terms in $\psi$ and $\pred{P}$ to derive feasible clauses for $\constr{\psi}$.
To illustrate, suppose that $(a[\ii] \le a[n]) \in \pred{P}$ and $\psi := a[p] \le a[n] \wedge a' = a\{i \leftarrow a[p]\}$.
Then, the first template in Figure~\ref{fig:constraint-templates} syntactically derives a feasible clause $\{ i=n,~a'[i] \le a'[n] \}$ for $\constr{\psi}$.
(To see this, note that the first template is in the form of $A \wedge B \wedge C \Rightarrow D$,
which is equivalent to $A \wedge B \Rightarrow \{\neg C, D\}$.
By instantiating $A$ to $a' = a\{i \leftarrow a[p]\}$, $B$ to $a[p] \le a[n]$, $C$ to $i \neq n$, and $D$ to $a'[i] \le a'[n]$, we obtain the desired clause.)
Although such clauses may also be produced using the incremental method described earlier,
we choose to generate them directly using templates as they are essential for computing precise transition formulae abstractions. 

When computing an abstraction formula, we include singleton feasible clauses by default,
and incrementally inject other nontrivial feasible clauses in the refinement loop.
For transition formulae, we furthermore include clauses populated by templates.
We convert the abstraction formulae of an array system specification %
into finite-state automata mechanically, leveraging Mona~\cite{klarlund2001mona} and Z3~\cite{Z3}.
These automata comprise an abstract regular transition system for the concrete specification.
We then analyze this abstract system using suitable regular model checkers in the literature.
For safety properties, we perform the abstract analysis
using the safety verifier by~\cite{chen2017learning},
which employs the $L^*$ learning algorithm to generate inductive invariants.
For liveness properties, we apply two independent, fully automated techniques:
one combines liveness-to-safety reduction techniques (see Section~\ref{sec:fairness})
with the safety verifier by Chen et~al.;
the other exploits the liveness verifier SLRP~\cite{lin2016liveness},
which essentially uses an SAT solver to synthesize well-founded relations.
Note that our case studies contain two liveness properties with fairness requirements
(Dijkstra $P_2$ and Chang-Roberts~$P_2$).
SLRP can only verify liveness under arbitrary schedulers and does not apply to these two properties.

We conducted the experiments on a laptop computer
with a 3.6GHz Intel i7 processor, a 16GB memory limit, and a 20-minute timeout.
Table~\ref{tab:array-systems-safety} presents the results.
For each of these properties,
our prototype tool can compute sufficiently precise regular abstractions
using constraints generated by no more than 10 clauses from
the incremental and template-based methods.
As for the abstract analysis,
the safety properties turn out to be relatively easy:
all of them can be proved by the $L^*$ safety verifier in seconds.
Regarding liveness properties, SLRP successfully proves all four liveness properties not requiring fair schedulers,
while the $L^*$ safety verifier proves only one of them, namely selection sort $P_3$, within the timeout.
This property can be verified without using any indexed predicates; the resulting abstract system thus degenerates into an integer program.
Reducing a liveness property to a safety property could be expensive:
this operation accounts for the majority of the runtime for verifying merge sort $P_1$,
and is likely the main bottleneck in the verification of Dijkstra $P_1$ and $P_3$ by $L^*$.
Interestingly, the reduction takes a relatively short time for Dijkstra $P_2$ and Chang-Roberts $P_2$,
and hence allows the safety verifier to prove them in time.
Since these two properties hold only under fairness requirements, they cannot be handled by SLRP.

%% file: related-work.tex
\section{Related work}
\label{sec:related}
There is a huge body of research on the verification of array programs and systems.
In this section, we provide some context for our work and discuss related work that has not yet been mentioned elsewhere in the paper.

Prior work exploiting predicate abstraction and interpolation exists in the context of verifying quantified inductive invariants for array programs
\cite{alberti2012safari,lahiri2007predicate,mcmillan2008quantified,jhala2007array,seghir2009abstraction,cimatti2016infinite}.
\emph{Indexed predicates}, which are essentially predicates containing free index variables,
were first introduced by Flanagan and Qadeer~\cite{flanagan2002predicate} to compose universally quantified inductive invariants.
Lahiri and Bryant later extended and formalized this notion in the framework of \emph{indexed predicate abstraction (IPA)} 
\cite{lahiri2004constructing,lahiri2004indexed,lahiri2007predicate}.
IPA encodes abstract state sets as propositional formulae over Boolean variables, and universally quantifies index variables in the predicates when performing concretization. Since post-images of abstract states are generally not computable in this setting, Lahiri and Bryant devised overapproximations of them using a quantifier instantiation heuristic, which enabled the authors to reduce abstract reachability analysis to solving quantified Boolean formulae. In comparison, our abstraction framework encodes abstract state sets as first-order formulae over word variables, and exploits regular languages to overapproximate sets and relations in the abstract domain. These apparatuses allow us to reason about quantified formulae and temporal properties beyond those considered by Lahiri and Bryant. On the other hand, our abstraction function induces infinite abstract systems, for which we employ infinite-state model checking techniques (i.e.~regular model checking~\cite{Parosh12,lin2022regular}).

For a class of array systems that are used to model \emph{multi-threaded programs},
specialized predicate abstraction techniques
have been developed to infer universally quantified inter-thread properties. %
Many techniques along this line have focused on symmetric systems
(i.e.~the system behaves correctly regardless of thread arrangement
\cite{donaldson2011symmetry,donaldson2012counterexample,pani2023thread,basler2009symbolic})
and monotonic systems
(i.e.~the system is equipped with a well-quasi-ordering 
\cite{ranise2010backward,alberti2012lazy,alberti2012safari}).
Some techniques further abstract the target system into a symmetric/monotonic system by
combining predication abstraction with some form of counter or monotonic abstraction
(e.g.~\cite{clarke2006environment,clarke2008proving,abdulla2010constrained,ganjei2016counting,kaiser2017lost}),
thereby improving the expressiveness and effectiveness of the abstraction methods.
Most of these techniques were designed for safety verification, and it is unclear whether they could be extended to handle liveness properties effectively.

Liveness verification of array systems is much more difficult
than safety verification, and therefore has relatively fewer automatic techniques and tool supports.
In the context of multi-threaded programs,
\emph{thread-modular analysis}
\cite{cook2007proving,malkis2007precise,popeea2012compositional,ketema2017termination,pani2021rely,pani2023thread}
is a popular verification methodology.
The analysis considers each thread in isolation
and overapproximates the behavior of the other threads
by assuming that their effects are passive or irrelevant.
Thread-modular analysis is therefore not suitable for verifying
liveness properties that require coordination among threads.
For example, in the case of Dijkstra's self-stabilizing algorithm,
proving that Process 1 eventually updates its local variable requires showing that the other processes collectively
make progress on updating their own variables.
To reason about liveness properties relying on thread coordination,
Farzan et al.~\cite{farzan2016proving} introduced \emph{well-founded proof spaces} and \emph{quantified predicate automata (QPAs)}
for parameterized verification. Like our approach, well-founded proof spaces take a language-theoretic view of termination
and reduce checking termination to showing the absence of lasso-shaped counterexample paths. 
The verification procedure provided by Farzan et al.~explores symmetric proof spaces (i.e.~proofs closed under permutation of thread identifiers).
Also, the procedure requires to iteratively check language inclusion of QPAs, which is itself an undecidable problem
and yet to have an effective implementation.

\emph{Ivy} \cite{mcmillan2018deductive,mcmillan2020ivy}
is a deductive verification system that offers
an automatic liveness-to-safety proof tactic for properties
specified in LTL and a decidable fragment of pure first-order logic
called \emph{Effectively Propositional Logic (EPR)}.
Specifically, Ivy reduces liveness to safety for infinite-state systems through
\emph{dynamic abstraction}~\cite{padon2021temporal,padon2017reducing},
which is an overapproximation of cycle detection by dynamically choosing a finite projection of an infinite path.
Despite the automatic reduction to safety, the resulting safety property still needs to be proved by the user.
Moreover, Ivy does not support theory reasoning inherently,
i.e., theories have to be encoded in EPR.
Notably, in Ivy's abstraction scheme, the concrete specification, the abstract system, and the safety proof are all expressed in EPR.
Our logical formalism is not limited to EPR and directly supports array logic with background element theories.
(Our presentation has used the theory of Difference Arithmetic for array elements,
as the theory already suffices to analyze our case studies.
But in principle, it is easy to adapt our techniques to other background theories such as Linear Integer Arithmetic.)

%% file: conclusion.tex
\section{Concluding remarks}
\label{sec:conc}

By combining indexed predicate abstraction, decision procedures, automatic structures, and regular model checking,
we present a novel framework to verify linear-time properties of array systems.
Given a first-order correctness specification,
our framework provides a systematic method to compute regular overapproximations of the specification as a string rewriting system, which allows us to exploit a wealth of regular model checking techniques for analyzing both safety and liveness properties.
Our experimental results show that this approach is able to verify non-trivial properties of array systems in several interesting case studies.

There are several immediate future research directions. Firstly, existing regular model checking techniques do not have good support for large alphabets. For this reason, the size of the resulting regular abstraction may grow exponentially in the number of predicates, especially with explicit representations of the letters (i.e.~bitvectors). This leads to the following question: {is it possible to extend regular model checking with symbolic representations of the alphabet symbols?} We believe that the answer to this question is positive, given promising work in symbolic automata learning~\cite{learning-symaut,learning-symaut2} and bitvector theory SMT solving \cite{yao2020fast,shi2021coqqfbv,peled2023smt}.
Secondly, our abstraction computation procedure currently uses simple incremental and template-based methods to search for constraints. We are confident that more sophisticated techniques such as the refinement-based search in IC3~\cite{ic3-array-bjorner,cimatti2014ic3} can be integrated with our procedure for constraint discovery.
Lastly, we have assumed in this work that appropriate indexed predicates have been supplied. 
Generating nontrivial indexed predicates for array properties is very challenging. In our case studies, we extract predicates from the atomic formulae of system and property specifications. 
Thus far, this suffices for our current goal, which has been to demonstrate the expressive power and viability of regular abstractions for array systems in interesting case studies. The next step of our research is, therefore, to study automatic generation and refinement of indexed predicates.

%% file: main-popl.bbl
\begin{thebibliography}{}

\end{thebibliography}



\begin{thebibliography}{94}


\ifx \showCODEN    \undefined \def \showCODEN     #1{\unskip}     \fi
\ifx \showDOI      \undefined \def \showDOI       #1{#1}\fi
\ifx \showISBNx    \undefined \def \showISBNx     #1{\unskip}     \fi
\ifx \showISBNxiii \undefined \def \showISBNxiii  #1{\unskip}     \fi
\ifx \showISSN     \undefined \def \showISSN      #1{\unskip}     \fi
\ifx \showLCCN     \undefined \def \showLCCN      #1{\unskip}     \fi
\ifx \shownote     \undefined \def \shownote      #1{#1}          \fi
\ifx \showarticletitle \undefined \def \showarticletitle #1{#1}   \fi
\ifx \showURL      \undefined \def \showURL       {\relax}        \fi
\providecommand\bibfield[2]{#2}
\providecommand\bibinfo[2]{#2}
\providecommand\natexlab[1]{#1}
\providecommand\showeprint[2][]{arXiv:#2}

\bibitem[Abadi(1989)]%
        {abadi1989power}
\bibfield{author}{\bibinfo{person}{Martin Abadi}.} \bibinfo{year}{1989}\natexlab{}.
\newblock \showarticletitle{The power of temporal proofs}.
\newblock \bibinfo{journal}{\emph{Theoretical Computer Science}} \bibinfo{volume}{65}, \bibinfo{number}{1} (\bibinfo{year}{1989}), \bibinfo{pages}{35--83}.
\newblock


\bibitem[Abdulla(2012)]%
        {Parosh12}
\bibfield{author}{\bibinfo{person}{Parosh~Aziz Abdulla}.} \bibinfo{year}{2012}\natexlab{}.
\newblock \showarticletitle{Regular Model Checking}.
\newblock \bibinfo{journal}{\emph{International Journal on Software Tools for Technology Transfer (STTT)}} \bibinfo{volume}{14}, \bibinfo{number}{2} (\bibinfo{year}{2012}), \bibinfo{pages}{109--118}.
\newblock


\bibitem[Abdulla et~al\mbox{.}(2010)]%
        {abdulla2010constrained}
\bibfield{author}{\bibinfo{person}{Parosh~Aziz Abdulla}, \bibinfo{person}{Yu-Fang Chen}, \bibinfo{person}{Giorgio Delzanno}, \bibinfo{person}{Fr{\'e}d{\'e}ric Haziza}, \bibinfo{person}{Chih-Duo Hong}, {and} \bibinfo{person}{Ahmed Rezine}.} \bibinfo{year}{2010}\natexlab{}.
\newblock \showarticletitle{Constrained Monotonic Abstraction: {A} {CEGAR} for Parameterized Verification}. In \bibinfo{booktitle}{\emph{International Conference on Concurrency Theory (CONCUR)}}. \bibinfo{publisher}{Springer}, \bibinfo{pages}{86--101}.
\newblock


\bibitem[Abdulla et~al\mbox{.}(2009)]%
        {abdulla2009approximated}
\bibfield{author}{\bibinfo{person}{Parosh~Aziz Abdulla}, \bibinfo{person}{Giorgio Delzanno}, {and} \bibinfo{person}{Ahmed Rezine}.} \bibinfo{year}{2009}\natexlab{}.
\newblock \showarticletitle{Approximated Parameterized Verification of Infinite-State Processes with Global Conditions}.
\newblock \bibinfo{journal}{\emph{Formal Methods in System Design (FMSD)}} \bibinfo{volume}{34}, \bibinfo{number}{2} (\bibinfo{year}{2009}), \bibinfo{pages}{126--156}.
\newblock


\bibitem[Abdulla et~al\mbox{.}(2016)]%
        {abdulla2016parameterized}
\bibfield{author}{\bibinfo{person}{Parosh~Aziz Abdulla}, \bibinfo{person}{Fr{\'e}d{\'e}ric Haziza}, {and} \bibinfo{person}{Luk{\'a}{\v{s}} Hol{\'\i}k}.} \bibinfo{year}{2016}\natexlab{}.
\newblock \showarticletitle{Parameterized Verification Through View Abstraction}.
\newblock \bibinfo{journal}{\emph{International Journal on Software Tools for Technology Transfer (STTT)}} \bibinfo{volume}{18}, \bibinfo{number}{5} (\bibinfo{year}{2016}), \bibinfo{pages}{495--516}.
\newblock


\bibitem[Abdulla et~al\mbox{.}(2012)]%
        {abdulla2012regular}
\bibfield{author}{\bibinfo{person}{Parosh~Aziz Abdulla}, \bibinfo{person}{Bengt Jonsson}, \bibinfo{person}{Marcus Nilsson}, \bibinfo{person}{Julien d'Orso}, {and} \bibinfo{person}{Mayank Saksena}.} \bibinfo{year}{2012}\natexlab{}.
\newblock \showarticletitle{Regular Model Checking for {LTL(MSO)}}.
\newblock \bibinfo{journal}{\emph{International Journal on Software Tools for Technology Transfer (STTT)}} \bibinfo{volume}{14}, \bibinfo{number}{2} (\bibinfo{year}{2012}), \bibinfo{pages}{223--241}.
\newblock


\bibitem[Alberti et~al\mbox{.}(2012a)]%
        {alberti2012lazy}
\bibfield{author}{\bibinfo{person}{Francesco Alberti}, \bibinfo{person}{Roberto Bruttomesso}, \bibinfo{person}{Silvio Ghilardi}, \bibinfo{person}{Silvio Ranise}, {and} \bibinfo{person}{Natasha Sharygina}.} \bibinfo{year}{2012}\natexlab{a}.
\newblock \showarticletitle{Lazy Abstraction with Interpolants for Arrays}. In \bibinfo{booktitle}{\emph{International Conference on Logic Programming and Automated Reasoning (LPAR)}}. \bibinfo{publisher}{Springer}, \bibinfo{pages}{46--61}.
\newblock


\bibitem[Alberti et~al\mbox{.}(2012b)]%
        {alberti2012safari}
\bibfield{author}{\bibinfo{person}{Francesco Alberti}, \bibinfo{person}{Roberto Bruttomesso}, \bibinfo{person}{Silvio Ghilardi}, \bibinfo{person}{Silvio Ranise}, {and} \bibinfo{person}{Natasha Sharygina}.} \bibinfo{year}{2012}\natexlab{b}.
\newblock \showarticletitle{SAFARI: {SMT}-Based Abstraction for Arrays with Interpolants}. In \bibinfo{booktitle}{\emph{International Conference on Computer-Aided Verification (CAV)}}. \bibinfo{publisher}{Springer}, \bibinfo{pages}{679--685}.
\newblock


\bibitem[Alberti et~al\mbox{.}(2017)]%
        {alberti2017framework}
\bibfield{author}{\bibinfo{person}{Francesco Alberti}, \bibinfo{person}{Silvio Ghilardi}, {and} \bibinfo{person}{Natasha Sharygina}.} \bibinfo{year}{2017}\natexlab{}.
\newblock \showarticletitle{A Framework for the Verification of Parameterized Infinite-State Systems}.
\newblock \bibinfo{journal}{\emph{Fundamenta Informaticae}} \bibinfo{volume}{150}, \bibinfo{number}{1} (\bibinfo{year}{2017}), \bibinfo{pages}{1--24}.
\newblock


\bibitem[Argyros and D'Antoni(2018)]%
        {learning-symaut2}
\bibfield{author}{\bibinfo{person}{George Argyros} {and} \bibinfo{person}{Loris D'Antoni}.} \bibinfo{year}{2018}\natexlab{}.
\newblock \showarticletitle{The Learnability of Symbolic Automata}. In \bibinfo{booktitle}{\emph{International Conference on Computer Aided Verification (CAV)}}. \bibinfo{publisher}{Springer}, \bibinfo{pages}{427--445}.
\newblock


\bibitem[Basler et~al\mbox{.}(2009)]%
        {basler2009symbolic}
\bibfield{author}{\bibinfo{person}{G{\'e}rard Basler}, \bibinfo{person}{Michele Mazzucchi}, \bibinfo{person}{Thomas Wahl}, {and} \bibinfo{person}{Daniel Kroening}.} \bibinfo{year}{2009}\natexlab{}.
\newblock \showarticletitle{Symbolic counter abstraction for concurrent software}. In \bibinfo{booktitle}{\emph{International Conference on Computer Aided Verification (CAV)}}. \bibinfo{publisher}{Springer}, \bibinfo{pages}{64--78}.
\newblock


\bibitem[Bend{\'\i}k and {\v{C}}ern{\'a}(2020)]%
        {bendik2020must}
\bibfield{author}{\bibinfo{person}{Jaroslav Bend{\'\i}k} {and} \bibinfo{person}{Ivana {\v{C}}ern{\'a}}.} \bibinfo{year}{2020}\natexlab{}.
\newblock \showarticletitle{{MUST}: minimal unsatisfiable subsets enumeration tool}. In \bibinfo{booktitle}{\emph{Tools and Algorithms for the Construction and Analysis of Systems (TACAS)}}. \bibinfo{publisher}{Springer}, \bibinfo{pages}{135--152}.
\newblock


\bibitem[Bend{\'\i}k et~al\mbox{.}(2018)]%
        {bendik2018recursive}
\bibfield{author}{\bibinfo{person}{Jaroslav Bend{\'\i}k}, \bibinfo{person}{Ivana {\v{C}}ern{\'a}}, {and} \bibinfo{person}{Nikola Bene{\v{s}}}.} \bibinfo{year}{2018}\natexlab{}.
\newblock \showarticletitle{Recursive online enumeration of all minimal unsatisfiable subsets}. In \bibinfo{booktitle}{\emph{Automated Technology for Verification and Analysis (ATVA)}}. \bibinfo{publisher}{Springer}, \bibinfo{pages}{143--159}.
\newblock


\bibitem[Benedikt et~al\mbox{.}(2003)]%
        {benedikt2003definable}
\bibfield{author}{\bibinfo{person}{Michael Benedikt}, \bibinfo{person}{Leonid Libkin}, \bibinfo{person}{Thomas Schwentick}, {and} \bibinfo{person}{Luc Segoufin}.} \bibinfo{year}{2003}\natexlab{}.
\newblock \showarticletitle{Definable Relations and First-Order Query Languages over Strings}.
\newblock \bibinfo{journal}{\emph{Journal of the ACM (JACM)}} \bibinfo{volume}{50}, \bibinfo{number}{5} (\bibinfo{year}{2003}), \bibinfo{pages}{694--751}.
\newblock


\bibitem[Bloem et~al\mbox{.}(2015)]%
        {sasha-book}
\bibfield{author}{\bibinfo{person}{Roderick Bloem}, \bibinfo{person}{Swen Jacobs}, \bibinfo{person}{Ayrat Khalimov}, \bibinfo{person}{Igor Konnov}, \bibinfo{person}{Sasha Rubin}, \bibinfo{person}{Helmut Veith}, {and} \bibinfo{person}{Josef Widder}.} \bibinfo{year}{2015}\natexlab{}.
\newblock \bibinfo{booktitle}{\emph{Decidability of Parameterized Verification}}.
\newblock \bibinfo{publisher}{Morgan {\&} Claypool Publishers}.
\newblock


\bibitem[Blumensath and Gradel(2000)]%
        {blumensath2000automatic}
\bibfield{author}{\bibinfo{person}{Achim Blumensath} {and} \bibinfo{person}{Erich Gradel}.} \bibinfo{year}{2000}\natexlab{}.
\newblock \showarticletitle{Automatic Structures}. In \bibinfo{booktitle}{\emph{Symposium on Logic in Computer Science (LICS)}}. \bibinfo{publisher}{IEEE}, \bibinfo{pages}{51--62}.
\newblock


\bibitem[Blumensath and Gr{\"a}del(2004)]%
        {blumensath2004finite}
\bibfield{author}{\bibinfo{person}{Achim Blumensath} {and} \bibinfo{person}{Erich Gr{\"a}del}.} \bibinfo{year}{2004}\natexlab{}.
\newblock \showarticletitle{Finite Presentations of Infinite Structures: Automata and Interpretations}.
\newblock \bibinfo{journal}{\emph{Theory of Computing Systems (TCS)}} \bibinfo{volume}{37}, \bibinfo{number}{6} (\bibinfo{year}{2004}), \bibinfo{pages}{641--674}.
\newblock


\bibitem[Bouajjani et~al\mbox{.}(2004)]%
        {bouajjani2004abstract}
\bibfield{author}{\bibinfo{person}{Ahmed Bouajjani}, \bibinfo{person}{Peter Habermehl}, {and} \bibinfo{person}{Tom{\'a}{\v{s}} Vojnar}.} \bibinfo{year}{2004}\natexlab{}.
\newblock \showarticletitle{Abstract Regular Model Checking}. In \bibinfo{booktitle}{\emph{International Conference on Computer-Aided Verification (CAV)}}. \bibinfo{publisher}{Springer}, \bibinfo{pages}{372--386}.
\newblock


\bibitem[Bouajjani et~al\mbox{.}(2000)]%
        {RMC}
\bibfield{author}{\bibinfo{person}{Ahmed Bouajjani}, \bibinfo{person}{Bengt Jonsson}, \bibinfo{person}{Marcus Nilsson}, {and} \bibinfo{person}{Tayssir Touili}.} \bibinfo{year}{2000}\natexlab{}.
\newblock \showarticletitle{Regular Model Checking}. In \bibinfo{booktitle}{\emph{International Conference on Computer-Aided Verification (CAV)}}. \bibinfo{publisher}{Springer}, \bibinfo{pages}{403--418}.
\newblock


\bibitem[Bradley and Manna(1998)]%
        {bradley-book}
\bibfield{author}{\bibinfo{person}{Aaron~R. Bradley} {and} \bibinfo{person}{Zohar Manna}.} \bibinfo{year}{1998}\natexlab{}.
\newblock \bibinfo{booktitle}{\emph{The Calculus of Computation: Decision Procedures with Applications to Verification}}.
\newblock \bibinfo{publisher}{Springer}.
\newblock


\bibitem[Bradley et~al\mbox{.}(2006)]%
        {bradley2006}
\bibfield{author}{\bibinfo{person}{Aaron~R. Bradley}, \bibinfo{person}{Zohar Manna}, {and} \bibinfo{person}{Henny~B. Sipma}.} \bibinfo{year}{2006}\natexlab{}.
\newblock \showarticletitle{What's Decidable about Arrays?}. In \bibinfo{booktitle}{\emph{International Conference on Verification, Model Checking, and Abstract Interpretation (VMCAI)}}. \bibinfo{publisher}{Springer}, \bibinfo{pages}{427--442}.
\newblock


\bibitem[Chang and Roberts(1979)]%
        {chang1979improved}
\bibfield{author}{\bibinfo{person}{Ernest Chang} {and} \bibinfo{person}{Rosemary Roberts}.} \bibinfo{year}{1979}\natexlab{}.
\newblock \showarticletitle{An Improved Algorithm for Decentralized Extrema-Finding in Circular Configurations of Processes}.
\newblock \bibinfo{journal}{\emph{Communications of the ACM (CACM)}} \bibinfo{volume}{22}, \bibinfo{number}{5} (\bibinfo{year}{1979}), \bibinfo{pages}{281--283}.
\newblock


\bibitem[Chen et~al\mbox{.}(2017)]%
        {chen2017learning}
\bibfield{author}{\bibinfo{person}{Yu-Fang Chen}, \bibinfo{person}{Chih-Duo Hong}, \bibinfo{person}{Anthony~W. Lin}, {and} \bibinfo{person}{Philipp R{\"{u}}mmer}.} \bibinfo{year}{2017}\natexlab{}.
\newblock \showarticletitle{Learning to Prove Safety over Parameterised Concurrent Systems}. In \bibinfo{booktitle}{\emph{International Conference on Formal Methods in Computer-Aided Design (FMCAD)}}. \bibinfo{publisher}{Springer}, \bibinfo{pages}{76--83}.
\newblock


\bibitem[Cimatti et~al\mbox{.}(2014)]%
        {cimatti2014ic3}
\bibfield{author}{\bibinfo{person}{Alessandro Cimatti}, \bibinfo{person}{Alberto Griggio}, \bibinfo{person}{Sergio Mover}, {and} \bibinfo{person}{Stefano Tonetta}.} \bibinfo{year}{2014}\natexlab{}.
\newblock \showarticletitle{IC3 Modulo Theories via Implicit Predicate Abstraction}. In \bibinfo{booktitle}{\emph{International Conference on Tools and Algorithms for the Construction and Analysis of Systems (TACAS)}}. \bibinfo{publisher}{Springer}, \bibinfo{pages}{46--61}.
\newblock


\bibitem[Cimatti et~al\mbox{.}(2016)]%
        {cimatti2016infinite}
\bibfield{author}{\bibinfo{person}{Alessandro Cimatti}, \bibinfo{person}{Alberto Griggio}, \bibinfo{person}{Sergio Mover}, {and} \bibinfo{person}{Stefano Tonetta}.} \bibinfo{year}{2016}\natexlab{}.
\newblock \showarticletitle{Infinite-state invariant checking with IC3 and predicate abstraction}.
\newblock \bibinfo{journal}{\emph{Formal Methods in System Design}}  \bibinfo{volume}{49} (\bibinfo{year}{2016}), \bibinfo{pages}{190--218}.
\newblock


\bibitem[Cimatti et~al\mbox{.}(2022)]%
        {cimatti2022verification}
\bibfield{author}{\bibinfo{person}{Alessandro Cimatti}, \bibinfo{person}{Alberto Griggio}, {and} \bibinfo{person}{Gianluca Redondi}.} \bibinfo{year}{2022}\natexlab{}.
\newblock \showarticletitle{Verification of SMT Systems with Quantifiers}. In \bibinfo{booktitle}{\emph{International Symposium on Automated Technology for Verification and Analysis (ATVA)}}. \bibinfo{publisher}{Springer}, \bibinfo{pages}{154--170}.
\newblock


\bibitem[Cimatti et~al\mbox{.}(2021)]%
        {cimatti2021universal}
\bibfield{author}{\bibinfo{person}{Alessandro Cimatti}, \bibinfo{person}{Alberto Griggio}, \bibinfo{person}{Gianluca Redondi}, {et~al\mbox{.}}} \bibinfo{year}{2021}\natexlab{}.
\newblock \showarticletitle{Universal Invariant Checking of Parametric Systems with Quantifier-free SMT Reasoning}. In \bibinfo{booktitle}{\emph{International Conference on Automated Deduction (CADE)}}. \bibinfo{publisher}{Springer}, \bibinfo{pages}{131--147}.
\newblock


\bibitem[Clarke et~al\mbox{.}(2006)]%
        {clarke2006environment}
\bibfield{author}{\bibinfo{person}{Edmund Clarke}, \bibinfo{person}{Muralidhar Talupur}, {and} \bibinfo{person}{Helmut Veith}.} \bibinfo{year}{2006}\natexlab{}.
\newblock \showarticletitle{Environment Abstraction for Parameterized Verification}. In \bibinfo{booktitle}{\emph{International Workshop on Verification, Model Checking, and Abstract Interpretation (VMCAI)}}. \bibinfo{publisher}{Springer}, \bibinfo{pages}{126--141}.
\newblock


\bibitem[Clarke et~al\mbox{.}(2008)]%
        {clarke2008proving}
\bibfield{author}{\bibinfo{person}{Edmund Clarke}, \bibinfo{person}{Murali Talupur}, {and} \bibinfo{person}{Helmut Veith}.} \bibinfo{year}{2008}\natexlab{}.
\newblock \showarticletitle{Proving Ptolemy Right: The Environment Abstraction Framework for Model Checking Concurrent Systems}. In \bibinfo{booktitle}{\emph{International Conference on Tools and Algorithms for the Construction and Analysis of Systems (TACAS)}}. \bibinfo{publisher}{Springer}, \bibinfo{pages}{33--47}.
\newblock


\bibitem[Clarke et~al\mbox{.}(1986)]%
        {clarke1986reasoning}
\bibfield{author}{\bibinfo{person}{Edmund~M. Clarke}, \bibinfo{person}{Orna Grumberg}, {and} \bibinfo{person}{Michael~C. Browne}.} \bibinfo{year}{1986}\natexlab{}.
\newblock \showarticletitle{Reasoning about Networks with Many Identical Finite State Processes}. In \bibinfo{booktitle}{\emph{Symposium on Principles of Distributed Computing (PODC)}}. \bibinfo{publisher}{ACM}, \bibinfo{pages}{240--248}.
\newblock


\bibitem[Clarke~Jr et~al\mbox{.}(2018)]%
        {clarke2018model}
\bibfield{author}{\bibinfo{person}{Edmund~M. Clarke~Jr}, \bibinfo{person}{Orna Grumberg}, \bibinfo{person}{Daniel Kroening}, \bibinfo{person}{Doron Peled}, {and} \bibinfo{person}{Helmut Veith}.} \bibinfo{year}{2018}\natexlab{}.
\newblock \bibinfo{booktitle}{\emph{Model Checking}}.
\newblock \bibinfo{publisher}{MIT press}.
\newblock


\bibitem[Colcombet and L{\"o}ding(2007)]%
        {colcombet2007transforming}
\bibfield{author}{\bibinfo{person}{Thomas Colcombet} {and} \bibinfo{person}{Christof L{\"o}ding}.} \bibinfo{year}{2007}\natexlab{}.
\newblock \showarticletitle{Transforming Structures by Set Interpretations}.
\newblock \bibinfo{journal}{\emph{Logical Methods in Computer Science (LMCS)}} \bibinfo{volume}{3}, \bibinfo{number}{2} (\bibinfo{year}{2007}), \bibinfo{pages}{paper--4}.
\newblock


\bibitem[Cook et~al\mbox{.}(2007)]%
        {cook2007proving}
\bibfield{author}{\bibinfo{person}{Byron Cook}, \bibinfo{person}{Andreas Podelski}, {and} \bibinfo{person}{Andrey Rybalchenko}.} \bibinfo{year}{2007}\natexlab{}.
\newblock \showarticletitle{Proving Thread Termination}. In \bibinfo{booktitle}{\emph{Programming Language Design and Implementation (PLDI)}}. \bibinfo{publisher}{ACM}, \bibinfo{pages}{320--330}.
\newblock


\bibitem[Cousot and Cousot(1977)]%
        {cousot1977abstract}
\bibfield{author}{\bibinfo{person}{Patrick Cousot} {and} \bibinfo{person}{Radhia Cousot}.} \bibinfo{year}{1977}\natexlab{}.
\newblock \showarticletitle{Abstract Interpretation: A Unified Lattice Model for Static Analysis of Programs by Construction or Approximation of Fixpoints}. In \bibinfo{booktitle}{\emph{Symposium on Principles of Programming Languages (POPL)}}. \bibinfo{publisher}{ACM}, \bibinfo{pages}{238--252}.
\newblock


\bibitem[Daniel et~al\mbox{.}(2016)]%
        {daniel2016infinite}
\bibfield{author}{\bibinfo{person}{Jakub Daniel}, \bibinfo{person}{Alessandro Cimatti}, \bibinfo{person}{Alberto Griggio}, \bibinfo{person}{Stefano Tonetta}, {and} \bibinfo{person}{Sergio Mover}.} \bibinfo{year}{2016}\natexlab{}.
\newblock \showarticletitle{Infinite-State Liveness-To-Safety via Implicit Abstraction and Well-Founded Relations}. In \bibinfo{booktitle}{\emph{International Conference on Computer-Aided Verification (CAV)}}. \bibinfo{publisher}{Springer}, \bibinfo{pages}{271--291}.
\newblock


\bibitem[de~Moura and Bj{\o}rner(2008)]%
        {Z3}
\bibfield{author}{\bibinfo{person}{Leonardo~Mendon\c{c}a de Moura} {and} \bibinfo{person}{Nikolaj Bj{\o}rner}.} \bibinfo{year}{2008}\natexlab{}.
\newblock \showarticletitle{Z3: An Efficient {SMT} Solver}. In \bibinfo{booktitle}{\emph{International Conference on Tools and Algorithms for the Construction and Analysis of Systems (TACAS)}}. \bibinfo{publisher}{Springer}, \bibinfo{pages}{337--340}.
\newblock


\bibitem[Demri et~al\mbox{.}(2016)]%
        {demri2016temporal}
\bibfield{author}{\bibinfo{person}{St{\'e}phane Demri}, \bibinfo{person}{Valentin Goranko}, {and} \bibinfo{person}{Martin Lange}.} \bibinfo{year}{2016}\natexlab{}.
\newblock \bibinfo{booktitle}{\emph{Temporal logics in computer science: finite-state systems}}. Vol.~\bibinfo{volume}{58}.
\newblock \bibinfo{publisher}{Cambridge University Press}.
\newblock


\bibitem[Dijkstra(1982)]%
        {dijkstra1982self}
\bibfield{author}{\bibinfo{person}{Edsger~W. Dijkstra}.} \bibinfo{year}{1982}\natexlab{}.
\newblock \showarticletitle{Self-Stabilization in Spite of Distributed Control}.
\newblock In \bibinfo{booktitle}{\emph{Selected Writings on Computing: A Personal Perspective}}. \bibinfo{publisher}{Springer}, \bibinfo{pages}{41--46}.
\newblock


\bibitem[Donaldson et~al\mbox{.}(2011)]%
        {donaldson2011symmetry}
\bibfield{author}{\bibinfo{person}{Alastair Donaldson}, \bibinfo{person}{Alexander Kaiser}, \bibinfo{person}{Daniel Kroening}, {and} \bibinfo{person}{Thomas Wahl}.} \bibinfo{year}{2011}\natexlab{}.
\newblock \showarticletitle{Symmetry-aware predicate abstraction for shared-variable concurrent programs}. In \bibinfo{booktitle}{\emph{International Conference on Computer-Aided Verification (CAV)}}. \bibinfo{publisher}{Springer}, \bibinfo{pages}{356--371}.
\newblock


\bibitem[Donaldson et~al\mbox{.}(2012)]%
        {donaldson2012counterexample}
\bibfield{author}{\bibinfo{person}{Alastair~F. Donaldson}, \bibinfo{person}{Alexander Kaiser}, \bibinfo{person}{Daniel Kroening}, \bibinfo{person}{Michael Tautschnig}, {and} \bibinfo{person}{Thomas Wahl}.} \bibinfo{year}{2012}\natexlab{}.
\newblock \showarticletitle{Counterexample-Guided Abstraction Refinement for Symmetric Concurrent Programs}.
\newblock \bibinfo{journal}{\emph{Formal Methods in System Design (FMSD)}} \bibinfo{volume}{41}, \bibinfo{number}{1} (\bibinfo{year}{2012}), \bibinfo{pages}{25--44}.
\newblock


\bibitem[Drews and D'Antoni(2017)]%
        {learning-symaut}
\bibfield{author}{\bibinfo{person}{Samuel Drews} {and} \bibinfo{person}{Loris D'Antoni}.} \bibinfo{year}{2017}\natexlab{}.
\newblock \showarticletitle{Learning Symbolic Automata}. In \bibinfo{booktitle}{\emph{Tools and Algorithms for the Construction and Analysis of Systems (TACAS)}}, Vol.~\bibinfo{volume}{10205}. \bibinfo{publisher}{Springer}, \bibinfo{pages}{173--189}.
\newblock


\bibitem[Emerson and Namjoshi(2003)]%
        {emerson2003reasoning}
\bibfield{author}{\bibinfo{person}{E.~Allen Emerson} {and} \bibinfo{person}{Kedar~S. Namjoshi}.} \bibinfo{year}{2003}\natexlab{}.
\newblock \showarticletitle{On Reasoning about Rings}.
\newblock \bibinfo{journal}{\emph{International Journal of Foundations of Computer Science}} \bibinfo{volume}{14}, \bibinfo{number}{04} (\bibinfo{year}{2003}), \bibinfo{pages}{527--549}.
\newblock


\bibitem[Esparza et~al\mbox{.}(2012)]%
        {esparza2012proving}
\bibfield{author}{\bibinfo{person}{Javier Esparza}, \bibinfo{person}{Andreas Gaiser}, {and} \bibinfo{person}{Stefan Kiefer}.} \bibinfo{year}{2012}\natexlab{}.
\newblock \showarticletitle{Proving termination of probabilistic programs using patterns}. In \bibinfo{booktitle}{\emph{International Conference on Computer-Aided Verification (CAV)}}. \bibinfo{publisher}{Springer}, \bibinfo{pages}{123--138}.
\newblock


\bibitem[Farzan et~al\mbox{.}(2016)]%
        {farzan2016proving}
\bibfield{author}{\bibinfo{person}{Azadeh Farzan}, \bibinfo{person}{Zachary Kincaid}, {and} \bibinfo{person}{Andreas Podelski}.} \bibinfo{year}{2016}\natexlab{}.
\newblock \showarticletitle{Proving Liveness of Parameterized Programs}. In \bibinfo{booktitle}{\emph{Symposium on Logic in Computer Science (LICS)}}. \bibinfo{publisher}{IEEE}, \bibinfo{pages}{1--12}.
\newblock


\bibitem[Fedyukovich et~al\mbox{.}(2019)]%
        {fedyukovich2019quantified}
\bibfield{author}{\bibinfo{person}{Grigory Fedyukovich}, \bibinfo{person}{Sumanth Prabhu}, \bibinfo{person}{Kumar Madhukar}, {and} \bibinfo{person}{Aarti Gupta}.} \bibinfo{year}{2019}\natexlab{}.
\newblock \showarticletitle{Quantified Invariants via Syntax-Guided Synthesis}. In \bibinfo{booktitle}{\emph{Computer Aided Verification (CAV)}}. \bibinfo{publisher}{Springer}, \bibinfo{pages}{259--277}.
\newblock


\bibitem[Felli et~al\mbox{.}(2021)]%
        {felli2021smt}
\bibfield{author}{\bibinfo{person}{Paolo Felli}, \bibinfo{person}{Alessandro Gianola}, {and} \bibinfo{person}{Marco Montali}.} \bibinfo{year}{2021}\natexlab{}.
\newblock \showarticletitle{SMT-based safety checking of parameterized multi-agent systems}. In \bibinfo{booktitle}{\emph{Proceedings of the AAAI Conference on Artificial Intelligence (AAAI)}}, Vol.~\bibinfo{volume}{35}. \bibinfo{publisher}{PKP Publishing}, \bibinfo{pages}{6321--6330}.
\newblock


\bibitem[Fiedor et~al\mbox{.}(2017)]%
        {fiedor2017lazy}
\bibfield{author}{\bibinfo{person}{Tom{\'a}{\v s} Fiedor}, \bibinfo{person}{Luk{\'a}{\v s} Hol{\'i}k}, \bibinfo{person}{Petr Jank{\r u}}, \bibinfo{person}{Ond{\v r}ej Leng{\'a}l}, {and} \bibinfo{person}{Tom{\'a}{\v s} Vojnar}.} \bibinfo{year}{2017}\natexlab{}.
\newblock \showarticletitle{Lazy Automata Techniques for {WS1S}}. In \bibinfo{booktitle}{\emph{International Conference on Tools and Algorithms for the Construction and Analysis of Systems (TACAS)}}. \bibinfo{publisher}{Springer}, \bibinfo{pages}{407--425}.
\newblock


\bibitem[Flanagan and Qadeer(2002)]%
        {flanagan2002predicate}
\bibfield{author}{\bibinfo{person}{Cormac Flanagan} {and} \bibinfo{person}{Shaz Qadeer}.} \bibinfo{year}{2002}\natexlab{}.
\newblock \showarticletitle{Predicate Abstraction for Software Verification}. In \bibinfo{booktitle}{\emph{Symposium on Principles of Programming Languages (POPL)}}. \bibinfo{publisher}{ACM}, \bibinfo{pages}{191--202}.
\newblock


\bibitem[Ganjei et~al\mbox{.}(2016)]%
        {ganjei2016counting}
\bibfield{author}{\bibinfo{person}{Zeinab Ganjei}, \bibinfo{person}{Ahmed Rezine}, \bibinfo{person}{Petru Eles}, {and} \bibinfo{person}{Zebo Peng}.} \bibinfo{year}{2016}\natexlab{}.
\newblock \showarticletitle{Counting dynamically synchronizing processes}.
\newblock \bibinfo{journal}{\emph{International Journal on Software Tools for Technology Transfer (STTT)}}  \bibinfo{volume}{18} (\bibinfo{year}{2016}), \bibinfo{pages}{517--534}.
\newblock


\bibitem[Ge and De~Moura(2009)]%
        {ge2009complete}
\bibfield{author}{\bibinfo{person}{Yeting Ge} {and} \bibinfo{person}{Leonardo De~Moura}.} \bibinfo{year}{2009}\natexlab{}.
\newblock \showarticletitle{Complete Instantiation for Quantified Formulas in Satisfiability Modulo Theories}. In \bibinfo{booktitle}{\emph{International Conference on Computer Aided Verification (CAV)}}. \bibinfo{publisher}{Springer}, \bibinfo{pages}{306--320}.
\newblock


\bibitem[German and Sistla(1992)]%
        {german1992reasoning}
\bibfield{author}{\bibinfo{person}{Steven~M. German} {and} \bibinfo{person}{A.~Prasad Sistla}.} \bibinfo{year}{1992}\natexlab{}.
\newblock \showarticletitle{Reasoning about Systems with Many Processes}.
\newblock \bibinfo{journal}{\emph{Journal of the ACM (JACM)}} \bibinfo{volume}{39}, \bibinfo{number}{3} (\bibinfo{year}{1992}), \bibinfo{pages}{675--735}.
\newblock


\bibitem[Ghilardi et~al\mbox{.}(2021)]%
        {ghilardi2021interpolation}
\bibfield{author}{\bibinfo{person}{Silvio Ghilardi}, \bibinfo{person}{Alessandro Gianola}, {and} \bibinfo{person}{Deepak Kapur}.} \bibinfo{year}{2021}\natexlab{}.
\newblock \showarticletitle{Interpolation and Amalgamation for Arrays with MaxDiff}. In \bibinfo{booktitle}{\emph{International Conference on Foundations of Software Science and Computation Structures (FoSSaCS)}}. \bibinfo{publisher}{Springer}, \bibinfo{pages}{268--288}.
\newblock


\bibitem[Gurfinkel et~al\mbox{.}(2016)]%
        {gurfinkel2016smt}
\bibfield{author}{\bibinfo{person}{Arie Gurfinkel}, \bibinfo{person}{Sharon Shoham}, {and} \bibinfo{person}{Yuri Meshman}.} \bibinfo{year}{2016}\natexlab{}.
\newblock \showarticletitle{{SMT}-Based Verification of Parameterized Systems}. In \bibinfo{booktitle}{\emph{International Symposium on Foundations of Software Engineering (FSE)}}. \bibinfo{publisher}{ACM}, \bibinfo{pages}{338--348}.
\newblock


\bibitem[Gurfinkel et~al\mbox{.}(2018)]%
        {gurfinkel2018quantifiers}
\bibfield{author}{\bibinfo{person}{Arie Gurfinkel}, \bibinfo{person}{Sharon Shoham}, {and} \bibinfo{person}{Yakir Vizel}.} \bibinfo{year}{2018}\natexlab{}.
\newblock \showarticletitle{Quantifiers on Demand}. In \bibinfo{booktitle}{\emph{Automated Technology for Verification and Analysis (ATVA)}}. \bibinfo{publisher}{Springer}, \bibinfo{pages}{248--266}.
\newblock


\bibitem[Habermehl et~al\mbox{.}(2008)]%
        {habermehl2008logic}
\bibfield{author}{\bibinfo{person}{Peter Habermehl}, \bibinfo{person}{Radu Iosif}, {and} \bibinfo{person}{Tom{\'a}{\v{s}} Vojnar}.} \bibinfo{year}{2008}\natexlab{}.
\newblock \showarticletitle{A Logic of Singly Indexed Arrays}. In \bibinfo{booktitle}{\emph{International Conference on Logic Programming and Automated Reasoning (LPAR)}}. \bibinfo{publisher}{Springer}, \bibinfo{pages}{558--573}.
\newblock


\bibitem[Hodkinson et~al\mbox{.}(2000)]%
        {hodkinson1999decidable}
\bibfield{author}{\bibinfo{person}{Ian Hodkinson}, \bibinfo{person}{Frank Wolter}, {and} \bibinfo{person}{Michael Zakharyaschev}.} \bibinfo{year}{2000}\natexlab{}.
\newblock \showarticletitle{Decidable Fragments of First-Order Temporal Logics}. In \bibinfo{booktitle}{\emph{Annals of Pure and Applied Logic}}. \bibinfo{publisher}{Elsevier}, \bibinfo{pages}{181--185}.
\newblock


\bibitem[Hoenicke and Schindler(2018)]%
        {hoenicke2018efficient}
\bibfield{author}{\bibinfo{person}{Jochen Hoenicke} {and} \bibinfo{person}{Tanja Schindler}.} \bibinfo{year}{2018}\natexlab{}.
\newblock \showarticletitle{Efficient interpolation for the theory of arrays}. In \bibinfo{booktitle}{\emph{International Joint Conference on Automated Reasoning (IJCAR)}}. \bibinfo{publisher}{Springer}, \bibinfo{pages}{549--565}.
\newblock


\bibitem[Hong(2022)]%
        {hong2022symbolic}
\bibfield{author}{\bibinfo{person}{Chih-Duo Hong}.} \bibinfo{year}{2022}\natexlab{}.
\newblock \emph{\bibinfo{title}{Symbolic techniques for parameterised verification}}.
\newblock \bibinfo{thesistype}{Ph.\,D. Dissertation}. \bibinfo{school}{University of Oxford}.
\newblock


\bibitem[Jhala and McMillan(2007)]%
        {jhala2007array}
\bibfield{author}{\bibinfo{person}{Ranjit Jhala} {and} \bibinfo{person}{Kenneth~L. McMillan}.} \bibinfo{year}{2007}\natexlab{}.
\newblock \showarticletitle{Array Abstractions from Proofs}. In \bibinfo{booktitle}{\emph{International Conference on Computer-Aided Verification (CAV)}}. \bibinfo{publisher}{Springer}, \bibinfo{pages}{193--206}.
\newblock


\bibitem[Jhala et~al\mbox{.}(2018)]%
        {jhala2018predicate}
\bibfield{author}{\bibinfo{person}{Ranjit Jhala}, \bibinfo{person}{Andreas Podelski}, {and} \bibinfo{person}{Andrey Rybalchenko}.} \bibinfo{year}{2018}\natexlab{}.
\newblock \showarticletitle{Predicate Abstraction for Program Verification: Safety and Termination}.
\newblock In \bibinfo{booktitle}{\emph{Handbook of Model Checking}}. \bibinfo{publisher}{Springer}, \bibinfo{pages}{447--491}.
\newblock


\bibitem[Kaiser et~al\mbox{.}(2017)]%
        {kaiser2017lost}
\bibfield{author}{\bibinfo{person}{Alexander Kaiser}, \bibinfo{person}{Daniel Kroening}, {and} \bibinfo{person}{Thomas Wahl}.} \bibinfo{year}{2017}\natexlab{}.
\newblock \showarticletitle{Lost in abstraction: Monotonicity in multi-threaded programs}.
\newblock \bibinfo{journal}{\emph{Information and Computation}}  \bibinfo{volume}{252} (\bibinfo{year}{2017}), \bibinfo{pages}{30--47}.
\newblock


\bibitem[Ketema and Donaldson(2017)]%
        {ketema2017termination}
\bibfield{author}{\bibinfo{person}{Jeroen Ketema} {and} \bibinfo{person}{Alastair~F. Donaldson}.} \bibinfo{year}{2017}\natexlab{}.
\newblock \showarticletitle{Termination Analysis for {GPU} Kernels}.
\newblock \bibinfo{journal}{\emph{Science of Computer Programming}}  \bibinfo{volume}{148} (\bibinfo{year}{2017}), \bibinfo{pages}{107--122}.
\newblock


\bibitem[Klarlund and M{\o}ller(2001)]%
        {klarlund2001mona}
\bibfield{author}{\bibinfo{person}{Nils Klarlund} {and} \bibinfo{person}{Anders M{\o}ller}.} \bibinfo{year}{2001}\natexlab{}.
\newblock \bibinfo{booktitle}{\emph{Mona Version 1.4: User Manual}}.
\newblock \bibinfo{publisher}{BRICS, Department of Computer Science, University of Aarhus Denmark}.
\newblock


\bibitem[Klarlund et~al\mbox{.}(2002)]%
        {klarlund2002mona}
\bibfield{author}{\bibinfo{person}{Nils Klarlund}, \bibinfo{person}{Anders M{\o}ller}, {and} \bibinfo{person}{Michael~I. Schwartzbach}.} \bibinfo{year}{2002}\natexlab{}.
\newblock \showarticletitle{{MONA} Implementation Secrets}.
\newblock \bibinfo{journal}{\emph{International Journal of Foundations of Computer Science}} \bibinfo{volume}{13}, \bibinfo{number}{04} (\bibinfo{year}{2002}), \bibinfo{pages}{571--586}.
\newblock


\bibitem[Komuravelli et~al\mbox{.}(2015)]%
        {ic3-array-bjorner}
\bibfield{author}{\bibinfo{person}{Anvesh Komuravelli}, \bibinfo{person}{Nikolaj~S. Bj{\o}rner}, \bibinfo{person}{Arie Gurfinkel}, {and} \bibinfo{person}{Kenneth~L. McMillan}.} \bibinfo{year}{2015}\natexlab{}.
\newblock \showarticletitle{Compositional Verification of Procedural Programs using {Horn} Clauses over Integers and Arrays}. In \bibinfo{booktitle}{\emph{Formal Methods in Computer-Aided Design (FMCAD)}}. \bibinfo{publisher}{{IEEE}}, \bibinfo{pages}{89--96}.
\newblock


\bibitem[Kroening and Strichman(2016)]%
        {kroening-book}
\bibfield{author}{\bibinfo{person}{Daniel Kroening} {and} \bibinfo{person}{Ofer Strichman}.} \bibinfo{year}{2016}\natexlab{}.
\newblock \bibinfo{booktitle}{\emph{Decision Procedures}}.
\newblock \bibinfo{publisher}{Springer}.
\newblock


\bibitem[Lahiri and Bryant(2004a)]%
        {lahiri2004constructing}
\bibfield{author}{\bibinfo{person}{Shuvendu~K. Lahiri} {and} \bibinfo{person}{Randal~E. Bryant}.} \bibinfo{year}{2004}\natexlab{a}.
\newblock \showarticletitle{Constructing Quantified Invariants via Predicate Abstraction}. In \bibinfo{booktitle}{\emph{International Conference on Verification, Model Checking, and Abstract Interpretation (VMCAI)}}. \bibinfo{publisher}{Springer}, \bibinfo{pages}{267--281}.
\newblock


\bibitem[Lahiri and Bryant(2004b)]%
        {lahiri2004indexed}
\bibfield{author}{\bibinfo{person}{Shuvendu~K. Lahiri} {and} \bibinfo{person}{Randal~E. Bryant}.} \bibinfo{year}{2004}\natexlab{b}.
\newblock \showarticletitle{Indexed Predicate Discovery for Unbounded System Verification}. In \bibinfo{booktitle}{\emph{International Conference on Computer-Aided Verification (CAV)}}. \bibinfo{publisher}{Springer}, \bibinfo{pages}{135--147}.
\newblock


\bibitem[Lahiri and Bryant(2007)]%
        {lahiri2007predicate}
\bibfield{author}{\bibinfo{person}{Shuvendu~K. Lahiri} {and} \bibinfo{person}{Randal~E. Bryant}.} \bibinfo{year}{2007}\natexlab{}.
\newblock \showarticletitle{Predicate Abstraction with Indexed Predicates}.
\newblock \bibinfo{journal}{\emph{ACM Transactions on Computational Logic (TOCL)}} \bibinfo{volume}{9}, \bibinfo{number}{1} (\bibinfo{year}{2007}), \bibinfo{pages}{4--es}.
\newblock


\bibitem[Liffiton et~al\mbox{.}(2016)]%
        {liffiton2016fast}
\bibfield{author}{\bibinfo{person}{Mark~H. Liffiton}, \bibinfo{person}{Alessandro Previti}, \bibinfo{person}{Ammar Malik}, {and} \bibinfo{person}{Joao Marques-Silva}.} \bibinfo{year}{2016}\natexlab{}.
\newblock \showarticletitle{Fast, flexible {MUS} enumeration}.
\newblock \bibinfo{journal}{\emph{Constraints}}  \bibinfo{volume}{21} (\bibinfo{year}{2016}), \bibinfo{pages}{223--250}.
\newblock


\bibitem[Lin and R\"{u}mmer(2016)]%
        {lin2016liveness}
\bibfield{author}{\bibinfo{person}{Anthony~W. Lin} {and} \bibinfo{person}{Philipp R\"{u}mmer}.} \bibinfo{year}{2016}\natexlab{}.
\newblock \showarticletitle{Liveness of Randomised Parameterised Systems under Arbitrary Schedulers}. In \bibinfo{booktitle}{\emph{International Conference on Computer-Aided Verification (CAV)}}. \bibinfo{publisher}{Springer}, \bibinfo{pages}{112--133}.
\newblock


\bibitem[Lin and R{\"u}mmer(2022)]%
        {lin2022regular}
\bibfield{author}{\bibinfo{person}{Anthony~W. Lin} {and} \bibinfo{person}{Philipp R{\"u}mmer}.} \bibinfo{year}{2022}\natexlab{}.
\newblock \showarticletitle{Regular model checking revisited}.
\newblock In \bibinfo{booktitle}{\emph{Model Checking, Synthesis, and Learning: Essays Dedicated to Bengt Jonsson on The Occasion of His 60th Birthday}}. \bibinfo{publisher}{Springer}, \bibinfo{pages}{97--114}.
\newblock


\bibitem[Ma et~al\mbox{.}(2019)]%
        {ma2019i4}
\bibfield{author}{\bibinfo{person}{Haojun Ma}, \bibinfo{person}{Aman Goel}, \bibinfo{person}{Jean-Baptiste Jeannin}, \bibinfo{person}{Manos Kapritsos}, \bibinfo{person}{Baris Kasikci}, {and} \bibinfo{person}{Karem~A Sakallah}.} \bibinfo{year}{2019}\natexlab{}.
\newblock \showarticletitle{I4: incremental inference of inductive invariants for verification of distributed protocols}. In \bibinfo{booktitle}{\emph{The Symposium on Operating Systems Principles (SOSP)}}. \bibinfo{publisher}{ACM}, \bibinfo{pages}{370--384}.
\newblock


\bibitem[Malkis et~al\mbox{.}(2007)]%
        {malkis2007precise}
\bibfield{author}{\bibinfo{person}{Alexander Malkis}, \bibinfo{person}{Andreas Podelski}, {and} \bibinfo{person}{Andrey Rybalchenko}.} \bibinfo{year}{2007}\natexlab{}.
\newblock \showarticletitle{Precise thread-modular verification}. In \bibinfo{booktitle}{\emph{International Static Analysis Symposium (SAS)}}. \bibinfo{publisher}{Springer}, \bibinfo{pages}{218--232}.
\newblock


\bibitem[Mann et~al\mbox{.}(2022)]%
        {mann2022counterexample}
\bibfield{author}{\bibinfo{person}{Makai Mann}, \bibinfo{person}{Ahmed Irfan}, \bibinfo{person}{Alberto Griggio}, \bibinfo{person}{Oded Padon}, {and} \bibinfo{person}{Clark Barrett}.} \bibinfo{year}{2022}\natexlab{}.
\newblock \showarticletitle{Counterexample-Guided Prophecy for Model Checking Modulo the Theory of Arrays}.
\newblock \bibinfo{journal}{\emph{Logical Methods in Computer Science (LMCS)}}  \bibinfo{volume}{18} (\bibinfo{year}{2022}), \bibinfo{pages}{131--147}.
\newblock


\bibitem[Manna and Pnueli(2012)]%
        {manna2012temporal}
\bibfield{author}{\bibinfo{person}{Zohar Manna} {and} \bibinfo{person}{Amir Pnueli}.} \bibinfo{year}{2012}\natexlab{}.
\newblock \bibinfo{booktitle}{\emph{The Temporal Logic of Reactive and Concurrent Systems: {S}pecification}}.
\newblock \bibinfo{publisher}{Springer Science \& Business Media}.
\newblock


\bibitem[McCarthy(1993)]%
        {mccarthy1993towards}
\bibfield{author}{\bibinfo{person}{John McCarthy}.} \bibinfo{year}{1993}\natexlab{}.
\newblock \showarticletitle{Towards a mathematical science of computation}.
\newblock \bibinfo{journal}{\emph{Program Verification: Fundamental Issues in Computer Science}} \bibinfo{volume}{1}, \bibinfo{number}{1} (\bibinfo{year}{1993}), \bibinfo{pages}{35--56}.
\newblock


\bibitem[McMillan(2008)]%
        {mcmillan2008quantified}
\bibfield{author}{\bibinfo{person}{Kenneth~L. McMillan}.} \bibinfo{year}{2008}\natexlab{}.
\newblock \showarticletitle{Quantified Invariant Generation Using an Interpolating Saturation Prover}. In \bibinfo{booktitle}{\emph{International Conference on Tools and Algorithms for the Construction and Analysis of Systems (TACAS)}}. \bibinfo{publisher}{Springer}, \bibinfo{pages}{413--427}.
\newblock


\bibitem[McMillan(2018)]%
        {mcmillan2018eager}
\bibfield{author}{\bibinfo{person}{Kenneth~L. McMillan}.} \bibinfo{year}{2018}\natexlab{}.
\newblock \showarticletitle{Eager Abstraction for Symbolic Model Checking}. In \bibinfo{booktitle}{\emph{International Conference on Computer Aided Verification (CAV)}}. \bibinfo{publisher}{Springer}, \bibinfo{pages}{191--208}.
\newblock


\bibitem[McMillan and Padon(2018)]%
        {mcmillan2018deductive}
\bibfield{author}{\bibinfo{person}{Kenneth~L. McMillan} {and} \bibinfo{person}{Oded Padon}.} \bibinfo{year}{2018}\natexlab{}.
\newblock \showarticletitle{Deductive Verification in Decidable Fragments with {Ivy}}. In \bibinfo{booktitle}{\emph{International Static Analysis Symposium (SAS)}}. \bibinfo{publisher}{Springer}, \bibinfo{pages}{43--55}.
\newblock


\bibitem[McMillan and Padon(2020)]%
        {mcmillan2020ivy}
\bibfield{author}{\bibinfo{person}{Kenneth~L. McMillan} {and} \bibinfo{person}{Oded Padon}.} \bibinfo{year}{2020}\natexlab{}.
\newblock \showarticletitle{Ivy: A Multi-Modal Verification Tool for Distributed Algorithms}. In \bibinfo{booktitle}{\emph{International Conference on Computer Aided Verification (CAV)}}. \bibinfo{publisher}{Springer}, \bibinfo{pages}{190--202}.
\newblock


\bibitem[Padon et~al\mbox{.}(2017)]%
        {padon2017reducing}
\bibfield{author}{\bibinfo{person}{Oded Padon}, \bibinfo{person}{Jochen Hoenicke}, \bibinfo{person}{Giuliano Losa}, \bibinfo{person}{Andreas Podelski}, \bibinfo{person}{Mooly Sagiv}, {and} \bibinfo{person}{Sharon Shoham}.} \bibinfo{year}{2017}\natexlab{}.
\newblock \showarticletitle{Reducing Liveness to Safety in First-Order Logic}.
\newblock \bibinfo{journal}{\emph{Symposium on Principles of Programming Languages (POPL)}} \bibinfo{volume}{2}, \bibinfo{number}{POPL} (\bibinfo{year}{2017}), \bibinfo{pages}{1--33}.
\newblock


\bibitem[Padon et~al\mbox{.}(2021)]%
        {padon2021temporal}
\bibfield{author}{\bibinfo{person}{Oded Padon}, \bibinfo{person}{Jochen Hoenicke}, \bibinfo{person}{Kenneth~L. McMillan}, \bibinfo{person}{Andreas Podelski}, \bibinfo{person}{Mooly Sagiv}, {and} \bibinfo{person}{Sharon Shoham}.} \bibinfo{year}{2021}\natexlab{}.
\newblock \showarticletitle{Temporal prophecy for proving temporal properties of infinite-state systems}.
\newblock \bibinfo{journal}{\emph{Formal Methods in System Design (FMSD)}}  \bibinfo{volume}{57} (\bibinfo{year}{2021}), \bibinfo{pages}{246--269}.
\newblock


\bibitem[Pani et~al\mbox{.}(2021)]%
        {pani2021rely}
\bibfield{author}{\bibinfo{person}{Thomas Pani}, \bibinfo{person}{Georg Weissenbacher}, {and} \bibinfo{person}{Florian Zuleger}.} \bibinfo{year}{2021}\natexlab{}.
\newblock \showarticletitle{Rely-guarantee bound analysis of parameterized concurrent shared-memory programs: With an application to proving that non-blocking algorithms are bounded lock-free}.
\newblock \bibinfo{journal}{\emph{Formal Methods in System Design (FMSD)}} \bibinfo{volume}{57}, \bibinfo{number}{2} (\bibinfo{year}{2021}), \bibinfo{pages}{270--302}.
\newblock


\bibitem[Pani et~al\mbox{.}(2023)]%
        {pani2023thread}
\bibfield{author}{\bibinfo{person}{Thomas Pani}, \bibinfo{person}{Georg Weissenbacher}, {and} \bibinfo{person}{Florian Zuleger}.} \bibinfo{year}{2023}\natexlab{}.
\newblock \showarticletitle{Thread-modular counter abstraction: automated safety and termination proofs of parameterized software by reduction to sequential program verification}.
\newblock \bibinfo{journal}{\emph{Formal Methods in System Design (FMSD)}}  \bibinfo{volume}{60} (\bibinfo{year}{2023}), \bibinfo{pages}{1--38}.
\newblock


\bibitem[Peled et~al\mbox{.}(2023)]%
        {peled2023smt}
\bibfield{author}{\bibinfo{person}{Matan~I Peled}, \bibinfo{person}{Bat-Chen Rothenberg}, {and} \bibinfo{person}{Shachar Itzhaky}.} \bibinfo{year}{2023}\natexlab{}.
\newblock \showarticletitle{SMT sampling via model-guided approximation}. In \bibinfo{booktitle}{\emph{International Symposium on Formal Methods (FM)}}. \bibinfo{publisher}{Springer}, \bibinfo{pages}{74--91}.
\newblock


\bibitem[Popeea and Rybalchenko(2012)]%
        {popeea2012compositional}
\bibfield{author}{\bibinfo{person}{Corneliu Popeea} {and} \bibinfo{person}{Andrey Rybalchenko}.} \bibinfo{year}{2012}\natexlab{}.
\newblock \showarticletitle{Compositional Termination Proofs for Multi-Threaded Programs}. In \bibinfo{booktitle}{\emph{International Conference on Tools and Algorithms for the Construction and Analysis of Systems (TACAS)}}. \bibinfo{publisher}{Springer}, \bibinfo{pages}{237--251}.
\newblock


\bibitem[Ranise and Ghilardi(2010)]%
        {ranise2010backward}
\bibfield{author}{\bibinfo{person}{Silvio Ranise} {and} \bibinfo{person}{Silvio Ghilardi}.} \bibinfo{year}{2010}\natexlab{}.
\newblock \showarticletitle{Backward Reachability of Array-Based Systems by {SMT} Solving: Termination and Invariant Synthesis}.
\newblock \bibinfo{journal}{\emph{Logical Methods in Computer Science (LMCS)}} \bibinfo{volume}{6}, \bibinfo{number}{4} (\bibinfo{year}{2010}), \bibinfo{pages}{25--44}.
\newblock


\bibitem[Schuppan and Biere(2006)]%
        {schuppan2006liveness}
\bibfield{author}{\bibinfo{person}{Viktor Schuppan} {and} \bibinfo{person}{Armin Biere}.} \bibinfo{year}{2006}\natexlab{}.
\newblock \showarticletitle{Liveness Checking as Safety Checking for Infinite State Spaces}.
\newblock \bibinfo{journal}{\emph{Electronic Notes in Theoretical Computer Science (ENTCS)}} \bibinfo{volume}{149}, \bibinfo{number}{1} (\bibinfo{year}{2006}), \bibinfo{pages}{79--96}.
\newblock


\bibitem[Seghir et~al\mbox{.}(2009)]%
        {seghir2009abstraction}
\bibfield{author}{\bibinfo{person}{Mohamed~Nassim Seghir}, \bibinfo{person}{Andreas Podelski}, {and} \bibinfo{person}{Thomas Wies}.} \bibinfo{year}{2009}\natexlab{}.
\newblock \showarticletitle{Abstraction Refinement for Quantified Array Assertions}. In \bibinfo{booktitle}{\emph{Static Analysis Symposium (SAS)}}. \bibinfo{publisher}{Springer}, \bibinfo{pages}{3--18}.
\newblock


\bibitem[Shi et~al\mbox{.}(2021)]%
        {shi2021coqqfbv}
\bibfield{author}{\bibinfo{person}{Xiaomu Shi}, \bibinfo{person}{Yu-Fu Fu}, \bibinfo{person}{Jiaxiang Liu}, \bibinfo{person}{Ming-Hsien Tsai}, \bibinfo{person}{Bow-Yaw Wang}, {and} \bibinfo{person}{Bo-Yin Yang}.} \bibinfo{year}{2021}\natexlab{}.
\newblock \showarticletitle{CoqQFBV: A Scalable Certified SMT Quantifier-Free Bit-Vector Solver}. In \bibinfo{booktitle}{\emph{International Conference on Computer Aided Verification (CAV)}}. \bibinfo{publisher}{Springer}, \bibinfo{pages}{149--171}.
\newblock


\bibitem[van Dalen(1994)]%
        {van1994logic}
\bibfield{author}{\bibinfo{person}{Dirk van Dalen}.} \bibinfo{year}{1994}\natexlab{}.
\newblock \bibinfo{booktitle}{\emph{Logic and structure}}. Vol.~\bibinfo{volume}{3}.
\newblock \bibinfo{publisher}{Springer}.
\newblock


\bibitem[Vardi and Wolper(1986)]%
        {vardi1986automata}
\bibfield{author}{\bibinfo{person}{Moshe~Y. Vardi} {and} \bibinfo{person}{Pierre Wolper}.} \bibinfo{year}{1986}\natexlab{}.
\newblock \showarticletitle{An Automata-Theoretic Approach to Automatic Program Verification}. In \bibinfo{booktitle}{\emph{Symposium on Logic in Computer Science (LICS)}}. \bibinfo{publisher}{IEEE}, \bibinfo{pages}{322--331}.
\newblock


\bibitem[Yao et~al\mbox{.}(2020)]%
        {yao2020fast}
\bibfield{author}{\bibinfo{person}{Peisen Yao}, \bibinfo{person}{Qingkai Shi}, \bibinfo{person}{Heqing Huang}, {and} \bibinfo{person}{Charles Zhang}.} \bibinfo{year}{2020}\natexlab{}.
\newblock \showarticletitle{Fast bit-vector satisfiability}. In \bibinfo{booktitle}{\emph{International Symposium on Software Testing and Analysis (ISSTA)}}. \bibinfo{publisher}{ACM}, \bibinfo{pages}{38--50}.
\newblock


\end{thebibliography}
